\documentclass[10pt,journal,compsoc]{IEEEtran}
\ifCLASSOPTIONcompsoc
  \usepackage[nocompress]{cite}
\else
  \usepackage{cite}
\fi

\usepackage{amsmath,amssymb,amsfonts}
\usepackage{algorithm}
\usepackage{algorithmicx}
\usepackage{algpseudocode}
\usepackage{multirow}
\usepackage{subfigure}
\usepackage{graphicx}
\usepackage{textcomp}
\usepackage{xcolor}
\usepackage{url}

\usepackage{amsthm}
\usepackage{balance}  

\def\BibTeX{{\rm B\kern-.05em{\sc i\kern-.025em b}\kern-.08em
    T\kern-.1667em\lower.7ex\hbox{E}\kern-.125emX}}
    
\newtheorem{definition}{Definition}
\newtheorem{theorem}{Theorem}
\newtheorem{lemma}{Lemma}

\def\eg {\emph{e.g}.} 
\def\ie {\emph{i.e}.}

\def\etal{\emph{et al}.}

\ifCLASSINFOpdf
\else
\fi
\begin{document}
%
\title{High Dimensional Similarity Search with Satellite System Graph: Efficiency, Scalability, and Unindexed Query Compatibility}
%
%
%
%

\author{Cong Fu, Changxu Wang, Deng Cai
\IEEEcompsocitemizethanks{\IEEEcompsocthanksitem
Cong Fu, Changxu Wang, and Deng Cai are with the State Key Laboratory of Computer-Aided Design (CAD) and Computer Graphics (CG), Zhejiang University, Hangzhou 310027, China. Cong Fu and Changxu Wang are also with the Alibaba Group, Beijing 100102, China.
(e-mail: \{fucong.fc, changxu.wcx\}@alibaba-inc.com; dengcai@gmail.com).
}}

%
%

\markboth{Journal of IEEE Transactions on Pattern Analysis and Machine Intelligence, ~Vol.~?, No.~?, January~2020}%
{Shell \MakeLowercase{\textit{et al.}}: Bare Demo of IEEEtran.cls for Computer Society Journals}
%



\IEEEtitleabstractindextext{%
\begin{abstract}
Approximate Nearest Neighbor Search (ANNS) in high dimensional space is essential in database and information retrieval. Recently, there has been a surge of interest in exploring efficient graph-based indices for the ANNS problem. Among them, Navigating Spreading-out Graph (NSG) provides fine theoretical analysis and achieves state-of-the-art performance. However, we find there are several limitations with NSG: 1) NSG has no theoretical guarantee on nearest neighbor search when the query is not indexed in the database; 2) NSG is too sparse which harms the search performance. In addition, NSG suffers from high indexing complexity. To address above problems, we propose the Satellite System Graphs (SSG) and a practical variant NSSG. Specifically, we propose a novel pruning strategy to produce SSGs from the complete graph. SSGs define a new family of MSNETs in which the out-edges of each node are distributed evenly in all directions. Each node in the graph builds effective connections to its neighborhood omnidirectionally, whereupon we derive SSG's excellent theoretical properties for both indexed and unindexed queries. We can adaptively adjust the sparsity of an SSG with a hyper-parameter to optimize the search performance. Further, NSSG is proposed to reduce the indexing complexity of the SSG for large-scale applications. Both theoretical and extensive experimental analysis are provided to demonstrate the strengths of the proposed approach over the existing representative algorithms. Our code has been released at \url{https://github.com/ZJULearning/SSG}.
\end{abstract}

\begin{IEEEkeywords}
Nearest neighbors, Similarity search, High dimension, Large-scale database
\end{IEEEkeywords}}

\maketitle

\IEEEdisplaynontitleabstractindextext

%
\IEEEpeerreviewmaketitle

\IEEEraisesectionheading{\section{Introduction}\label{introduction}}
\IEEEPARstart{A}{pproximate} Nearest Neighbor Search (ANNS) has been a fundamental problem over decades and supports many applications in database, information retrieval, data mining, and machine learning \cite{BeisL97Shape,ferhatosmanoglu2001approximate, chen2005robust, philbin2007object, LiuRR07Clustering, zheng2016lazylsh, AroraSK018}. When machine learning, especially deep learning techniques are applied to more and more traditional large-scale applications, indexing and searching on dense-real-vector databases becomes a significant challenge. Due to the intrinsic difficulty of the exact nearest neighbor search, various solutions have been proposed to solve the Approximate Nearest Neighbor Search (ANNS) problem. For example, the tree-based methods \cite{Bentley1975Multidimensional, Fukunaga1975A, Silpaanan2008Optimised, Jagadish2005iDistance, Fu2000Dynamic}, the hashing-based methods \cite{Gionis1999Similarity, Weiss2008Spectral, Huang2015Query, Liu2016Query},  the quantization-based methods \cite{weber1998quantitative, jegou2011product, ge2014optimized, zhang2014composite, johnson2017billion} and the graph-based methods \cite{arya1993approximate, Hajebi2011Fast, malkov2014approximate, MalkovYHNSW16, Ben2016Fanng, fu2019fast}. Among them, the graph-based methods have shown promising search performance on widely used public datasets \cite{Hajebi2011Fast, MalkovYHNSW16, Ben2016Fanng, fu2019fast}  and been in the leading position (see the well-known benchmark \cite{aumuller2017ann}).

Graph-based methods build (proximity) graphs on the dataset as their indices for similarity retrieval. Their indices are typically a set of nodes and edges implemented by 2d-array or adjacent lists in procedures. The search algorithms used in graph-based methods are usually an A*-search like algorithm (given in Alg. \ref{search_alg}) or its variants. From a random or pre-selected fixed starting node, they hope to check the neighbors and neighbors' neighbors iteratively to locate a closer position to the query in the graph and move towards the query. Different from tree, hashing, and quantization, which try to solve ANNS mainly by partitioning the space, graph-based methods mainly benefit from the idea of ``connecting'' due to the nature of the routing-based search. 

As a consequence, the way to connect the nodes will influence the search performance significantly (see the recent survey \cite{shimomura2020survey} for more details). Also, the complexity of this algorithm is difficult to analyze because it is basically a greedy method. However, there are works like Monotonic Search Networks (MSNET)\cite{dearholt1988monotonic}, Randomized Neighborhood Graph (RNG*)\cite{arya1993approximate}, and Monotonic Relative Neighborhood Graph (MRNG)\cite{fu2019fast}, which make it analyzable based on their carefully designed graph structures. This is because their graphs enable a property called search-monotonicity. Specifically, for any query node $q$ and any search-starting node $s$, it is guaranteed to find a path from $s$ to $q$, $\{s(n_0),n_1,n_2,...,n_L,q\}$, on a search-monotonic graph with Alg. \ref{search_alg}. More importantly, $\forall 0\le l< L, \delta(n_l,q) > \delta(n_{l+1},q)$, where $\delta(n_l,q)$ is the distance between $n_l$ and $q$. In other words, a closer node to the query will always be found at next iteration. Based on the monotonicity of search, these graphs ensure a low theoretical search time complexity. In the resent work \cite{fu2019fast}, Fu \etal~ extend the monotonic-search-graph theory based on MSNET\cite{dearholt1988monotonic} and RNG*\cite{arya1993approximate}. Using similar pruning strategy with HNSW\cite{MalkovYHNSW16} and FANNG\cite{Ben2016Fanng}, they propose a practical algorithm NSG for large-scale scenarios with lower space and indexing-time consumption and better search performance. They demonstrate their strengths over many representative methods on public and large-scale datasets. Despite the success of NSG, there are three limitations in it:


\smallskip
\noindent
i. \textbf{NSG indexing algorithm produces an over-sparse graph, which is the bottleneck of their search performance.} NSG can be regarded as a graph pruned from a $k$ nearest neighbor graph (KNNG). Fu \etal~\cite{fu2019fast} try to sparsify the graph to accelerate the retrieval, but we argue that sparser graphs are not necessarily better.


\begin{algorithm}[t]\small
	\caption{Search-on-Graph($G$, $\textbf{p}$, $\textbf{q}$, $l$)}
	\label{search_alg}
	\begin{algorithmic}[1]
		\Require graph $G$, start node $\textbf{p}$, query point $\textbf{q}$, candidate pool size $l$
		\Ensure $k$ nearest neighbors of $\textbf{q}$
		\State $i$=0, candidate pool $S = \emptyset$
		\State $S$.add($\textbf{p}$)
		\While{$i < l$}
		\State $i=$the id of the first unchecked node $p_i$ in $S$
		\State mark $\textbf{p}_\textbf{i}$ as checked
		\ForAll {neighbor $\textbf{n}$ of $\textbf{p}_\textbf{i}$ in $G$}
		\If{$\textbf{n}$ has not been visited}
		\State $S$.add($\textbf{n}$)
		\EndIf
		\EndFor
		\State sort $S$ in ascending order of the distance to $\textbf{q}$
		\If{$S$.size() $>l$} 
		\State $S$.resize($l$) // remove nodes from back of $S$ to keep its\State size no larger than $l$
		\EndIf
		\EndWhile
		\State return the first $k$ nodes in $S$
	\end{algorithmic}
\end{algorithm}


\smallskip
\noindent
ii. \textbf{The theoretical properties of NSG is derived from an unstated assumption that the query is indexed in the graph.} Search for indexed and unindexed queries are seldom distinguished in most prior works. There may be few differences in tree, hashing, and quantization based methods (therefore also neglected in previous graph-based methods), but we will show that search on a graph behaves quite differently when the query is not indexed.



\smallskip
\noindent
iii. \textbf{The high time complexity of NSG edge selection limits its scalability,} which is the key part of their indexing.


To address above problems, we propose the Satellite System Graph. From our perspective, the search process on a graph index is very similar to the message transferring in a Communication Satellite System. In such a system, the neighbors of each satellite are distributed uniformly around. Consequently, the information propagation is very efficient in any direction, no matter where the target receivers are. Imitating satellites, SSG is a carefully designed graph structure, where the out-edges around each node are distributed evenly in all directions. Meanwhile, we prove SSGs belong to the MSNET family. Therefore, SSG not only inherits the excellent ANNS properties for indexed queries, but also has a good theoretical guarantee for unindexed queries. Further, the contributions of this paper is highlighted as follows:

\smallskip
\noindent
\textbf{1)} We propose a novel pruning strategy to produce a sparse graph SSG from an approximate KNNG for efficient similarity retrieval.

\smallskip
\noindent
\textbf{2)} We reveal the excellent theoretical properties of SSG on both indexed and unindexed queries. The sparsity of SSG can be adjusted adaptively with a hyper-parameter for performance optimization. 

\smallskip
\noindent
\textbf{3)} To reduce SSG's indexing complexity, we propose a practical variant NSSG based on SSG for large-scale search. Extensive experiments demonstrate the strengths of NSSG.

\section{Background}
\label{related_work}
\subsection{Notations}
Let $S$ denote a finite dataset with $n$ points. Let $E$ denote the Euclidean space with dimension $d$, commonly used in this literature. Let $\delta(\cdot,\cdot)$ denote the Euclidean distance function. Let $G$ denote a graph defined on $S$. Let $\overset{\longrightarrow}{pq}$ denote a directed edge from point $p$ to $q$. $Cone(\overset{\longrightarrow}{pq}$ $,\alpha)$ denotes a circular cone centered at $\overset{\longrightarrow}{pq}$ with angular diameter $2\alpha$. $B(p,\delta(p,q))$ denotes the open sphere centered at $p$ with radius $\delta(p,q)$.
\subsection{From NNS to ANNS} 
Due to the intrinsic difficulty of exact Nearest Neighbor Search (NNS), most researchers turn to ANNS. The main motivation is to trade a little loss in accuracy for much shorter search time. The formal definition is as follows \cite{Gionis1999Similarity}.
\begin{definition}[Nearest Neighbor Search]
    Given $S$ in $E$, preprocess $S$ to efficiently return $p \in S$ which is closest to $q$.
\end{definition}
This naturally generalizes to the $k-$NNS when we require the algorithm to return $k$ points ($k>1$) which are the closest to $q$. The approximate version of the NNS problem (ANNS) can be defined as follows.
\begin{definition}[$\epsilon-$Nearest Neighbor Search]
    Given $S$ in $E$, preprocess $S$ to efficiently return $p \in S$ such that $\delta(p,q) \le (1 + \epsilon)\delta(r, q)$, where $r$ is the nearest neighbor of $q \in S$, $\epsilon>0$.
\end{definition}

Similarly, this generalizes to the Approximate $k$ Nearest Neighbor Search (AKNNS) when we require the algorithm to return $k$ points ($k>1$) such that
$\forall i = 1, ..., k, \delta(p_i, q) \le (1 + \epsilon)\delta(r, q)$.


\subsection{Non-Graph Based Methods}
For the past decades, various methods are proposed to solve the AKNNS problem efficiently, including hashing-based, tree-based, quantization-based, and graph-based methods. The hashing-based methods try to split the space with hyper-surfaces and organize the dataset with hashing tables. Typical methods include Locality Sensitive Hashing (LSH) \cite{Gionis1999Similarity} and Spectral Hashing \cite{Weiss2008Spectral}. Tree-based methods try to partition the space into sub-regions and index them into tree structures. Representative methods include Randomized KD-Tree \cite{Silpaanan2008Optimised} and R-Tree \cite{beckmann1990r}. Quantization-based methods try to solve the AKNNS problem through reducing the complexity of distance calculations. Specifically, the algorithm quantizes the original data points and represents them as binary codes, which serve as the references to the codebook (quantizers). The complexity of the distance computation can be reduced significantly by computing the approximate distance with the pre-built codebook (quantizers). From a different angle, quantization can also be seen as space-splitting methods. Product Quantization (PQ) \cite{jegou2011product} and Composite Quantization\cite{zhang2014composite} are two typical methods. 

Recently, there are some works focusing on improving the performance of these three types methods in the methodology or engineering level such as \cite{HuangFZFN15, AroraSK018, ge2013optimized}. However, in the famous benchmark in this community \cite{aumuller2017ann} and many works \cite{Hajebi2011Fast, Jin2014Fast, CongEfanna2016, MalkovPLK14, MalkovYHNSW16, Ben2016Fanng, fu2019fast}, the performance of non-graph-based methods are left far behind graph-based ones under fair comparison. Fu \etal~\cite{fu2019fast} try to explain this by their empirical evaluation that NSG scans less than one tenth points to reach the same accuracy.
\begin{figure}[t]
\begin{center}
\includegraphics[width=240pt]{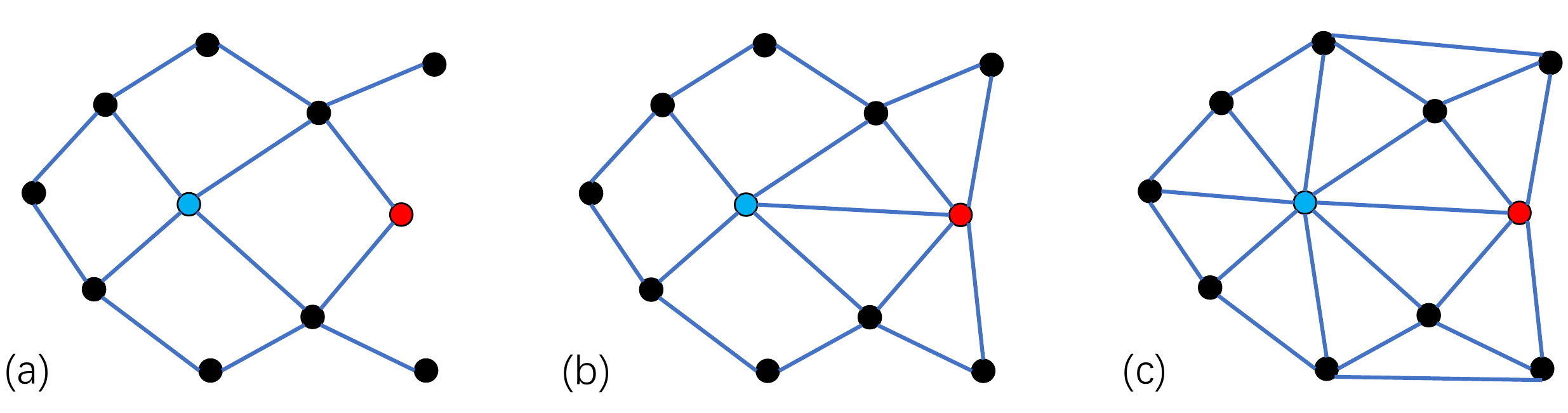}
\end{center}
   \caption{\textbf{An illustration of the influences of sparsity on search complexity} (a), (b), and (c) are three MSNETs on a toy dataset with different sparsity. In particular, (a) is built according to NSG's edge pruning strategy. Let the blue node be the search starting point and the red as the query. In (a) we need two hops via a neighbor of the blue node to reach the red, while in (b) and (c) we only need one hop. As a result, we need 7, 5, and 8 distance calculations to reach the answer on (a), (b), and (c) respectively. (b) delivers the best search performance.}
\label{msnetangle}
\end{figure}

\subsection{Graph-Based Methods.}
Recently, the graph-based methods have attracted wide interest and shown exciting results\cite{aumuller2017ann, shimomura2020survey}. Although various graph indices have been developed, they all use similar search algorithms as shown in Alg. \ref{search_alg}. The main idea of this A*-search-like algorithm is to iteratively discover the nodes which are closer to the query, among the neighbors of the current node. Thus, no matter how complex the graph is, the search time complexity can be roughly decomposed as $ol$ \cite{fu2019fast}, where $o$ is the out-degree of the graph and $l$ is the number of iterations of Alg. \ref{search_alg} (or the length of the ``search path''). From this perspective, the development of the graph indices can be summarized as two aspects: sparsify the graph and reduce the search path lengths.

There are several early graph models with excellent theoretical guarantees (e.g., the Delaunay Graphs \cite{aurenhammer1991voronoi, lee1980two} and the MSNET \cite{dearholt1988monotonic}) or empirical conclusions (\cite{kleinberg2000navigation, boguna2009navigability}) on the search path length. However, these graphs suffer from high indexing complexity or are not well-designed models for the ANNS problem. Recent works are mostly optimized approximations or variants of above structures. Specifically, KNN graph based methods (GNNS \cite{Hajebi2011Fast}, IEH \cite{Jin2014Fast}, Efanna \cite{CongEfanna2016}) stem from the Delaunay Graph. They are designed to reduce the out-degree of the Delaunay Graph in the high dimensions and maintain a considerably short search path; Navigating Small World Graph (NSW) \cite{malkov2014approximate} approximates the Navigating Small World Network \cite{kleinberg2000navigation} and is modified to adapt to high dimensions. Hierarchical Navigating Small World Graph (HNSW) \cite{MalkovYHNSW16} further improves the NSW by stacking multiple NSWs of different scopes, which intuitively shortens the search paths via different short-cuts on different graph layers; FANNG \cite{Ben2016Fanng} and HNSW \cite{MalkovYHNSW16} use similar edge-selection strategies as the Relative Neighborhood Graph (RNG)\cite{jaromczyk1992relative} and RNG*(S)\cite{arya1993approximate} to sparsify the graphs; Based on the MSNET \cite{dearholt1988monotonic}, RNG*\cite{arya1993approximate} and heuristic works HNSW and FANNG, Fu~\etal~develop the MRNG \cite{fu2019fast} with a good out-degree upper bound and a short search path guarantee. Further, they propose the NSG as a variant of the MRNG to reduce the indexing complexity. In their experimental study of several graph-based papers \cite{Hajebi2011Fast, Jin2014Fast, CongEfanna2016, MalkovPLK14, MalkovYHNSW16, Ben2016Fanng, fu2019fast}, the graph-based approaches outperform the non-graph based ones significantly on several frequently-used public datasets. Among them, NSG \cite{fu2019fast} is in the leading position to the best of our knowledge.

\begin{figure}[t]
\begin{center}
\includegraphics[width=220pt]{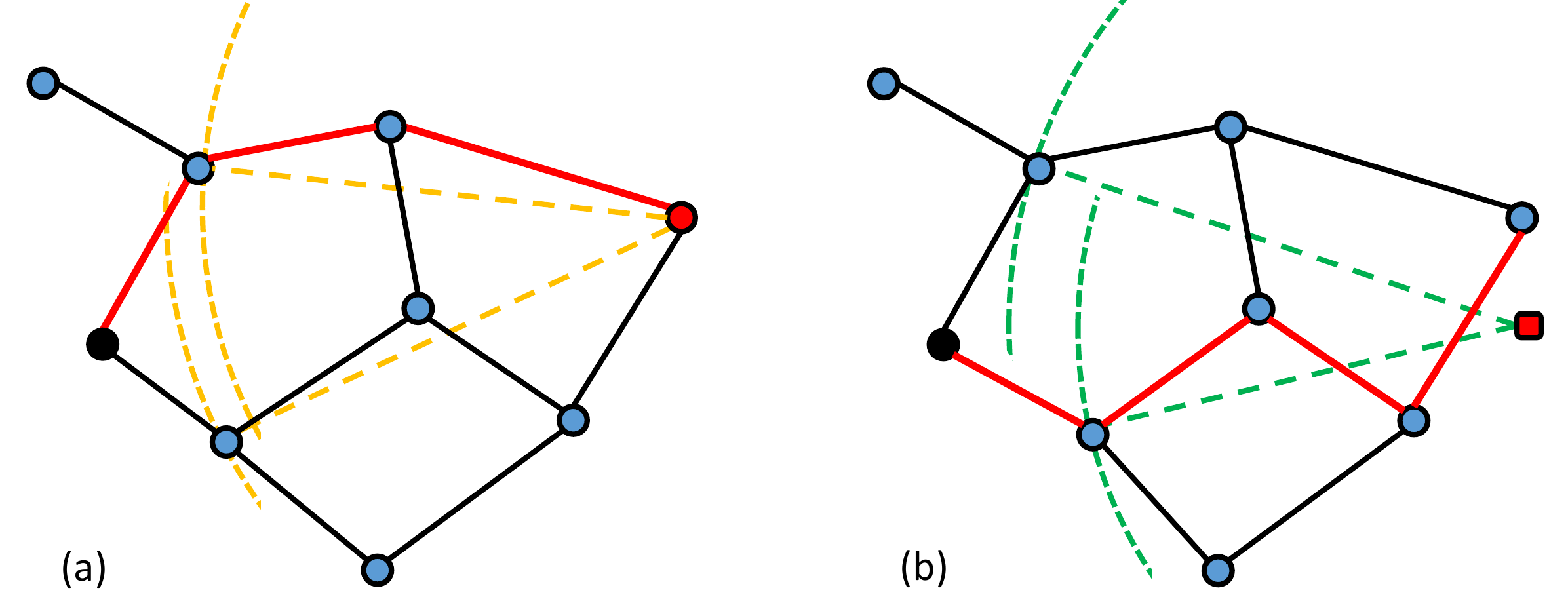}
\end{center}
   \caption{\textbf{An illustration of the difference between searching for an indexed query and an unindexed query.} The graph is built according to NSG's edge pruning strategy. The black node denotes the search-start node, and the red node denotes the query node. (a) shows the search for an indexed query, while (b) shows the search for the NN of an unindexed query. The distance to the query determines the choice of the first step. Therefore, though they arrive at the same location finally, the routes are completely different (red lines).}
\label{in-out-db-query}
\end{figure}

\subsection{Closely Related Works}
Most existing graph-based methods can be viewed as pruning edges from the complete graph or a KNNG. Different pruning strategy leads to different properties. This work and a closely related prior work, NSG\cite{fu2019fast}, is mainly built upon MSNET\cite{dearholt1988monotonic} due to MSNET's fine theoretical properties. We will briefly introduce MSNET, NSG, and other two similar works, RNG*\cite{arya1993approximate} and DPG\cite{li2016approximate} in this section.

\noindent
\textbf{i. MSNET.} As introduced in Sec. 1, MSNET \cite{dearholt1988monotonic} guarantees a ``monotonic" search process (with Alg. \ref{search_alg}) towards the query. Fu \etal~\cite{fu2019fast} prove that the search complexity on an MSNET grows in $O(n^{\frac{1}{d}}\log n /\triangle{r})$ regarding $n$, where $\triangle{r}$ can be treated as a constant in practice. Dearholt~\etal propose a minimal MSNET with appealing search performance. They build an RNG first and add edges to the RNG greedily (seek shortest edges meeting the monotonicity) until the whole graph ensures monotonicity, but this indexing is in high-order polynomial complexity.

\smallskip
\noindent
\textbf{ii. RNG* and RNG*(S)} Arya~\etal \cite{arya1993approximate} propose Randomized Neighborhood Graph (RNG*) with a theoretical search complexity of $O(\log^3n)$. RNG* is built by partitioning the space around a given point with multiple cones and select neighbors in each cone with a randomized greedy procedure, i.e., randomly permute (index) all potential neighbors in this cone and seek $O(\log n)$) approximately closest neighbors. Then this procedure is applied for each point. Though theoretically attractive, RNG* is of high space complexity (due to high degree), high indexing time complexity, and high search complexity in high dimensions. Therefore they propose a sparse variant RNG*(S), which can be seen as sparsifying the complete graph with a variant of RNG's pruning strategy. In addition, they also use auxiliary structure, KD-tree, to select search starting point to further accelerate the search. Still, RNG*(S) needs $O(n^3)$ time complexity for indexing.


\begin{algorithm}[t]\small
	\caption{SSG-Build($\textbf{D}$, $\textbf{a}$)}
	\label{ssgbuild_alg}
	\begin{algorithmic}[1]
		\Require dataset $\textbf{D}$, angle $\textbf{a}$
		\Ensure SSG $\textbf{G}$
		\ForAll {node $d_i$ in $\textbf{D}$}
			\State $L=\emptyset$
			\ForAll {node $d_j$ in $\textbf{D}-\{d_i\}$}
				\State $l_{ij}$ = length $d_id_j$
				\State $L$.add(($d_j$, $l_{ij}$))
			\EndFor
			\State $L$.sort() // ascending order of $l_{ij}$.
			\State $P$=$\{L[0]\}$
			\ForAll {node $(d_j, l_{ij})$ in $\textbf{L}$}
				\State flag=True
				\ForAll {node $(d_k, l_{jk})$ in $\textbf{P}$}
					\If {$cos \angle{d_jd_id_k} > cos \angle{a}$}
					\State flag = False
					\EndIf
				\EndFor
				\If {flag = True}
				\State P.add($(d_j, l_{ij})$)
				\EndIf
			\EndFor
			\State $G[d_i] = P$
		\EndFor
				
	\end{algorithmic}
\end{algorithm}	

\smallskip
\noindent
\textbf{iii. RNG*(S), MRNG and NSG.} MRNG and RNG*(S) are basically the same, which can also be viewed as a sparsified graph from the complete graph with a variant of RNG pruning. RNG is proved to be nonmonotonic\cite{dearholt1988monotonic}, while Arya~\etal prove RNG*(S) is monotonic\cite{arya1993approximate}. In an  RNG*(S), a point $r$ can join $p$'s neighbors if and only if $pr$ is not the longest edge in any existing triangle. Fu~\etal \cite{fu2019fast} extend MSNET\cite{dearholt1988monotonic} and RNG*(S) theory with search complexity and call it MRNG. Similarly, due to MRNG's high indexing cost ($O(n^2\log n)$), they propose a variant, NSG. NSG is built based on a pre-built approximate KNNG, then the edge selection is performed through a search-and-prune procedure. In addition, NSG uses a fix search starting point (navigating node), and they design a DFS-like routing to ensure a unidirectional path from the navigating node to each other node for connectivity.


\smallskip
\noindent
\textbf{iv. DPG.} Similar to this work, the Diversified Proximity Graph (DPG) \cite{li2016approximate} uses angles as a key element in their pruning criteria. They first build a KNNG, and then prune half the edges down to maximize the average angle between edges. Further, they add opposite edges for all unidirectional edges to get an undirected graph.

\smallskip
\noindent
\textbf{v. Differences regarding this work.} 1) We will prove SSGs belong to the MSNET family later, and SSGs are different from the minimal MSNET because we do not think sparser graphs are necessarily better. On the contrary, SSGs are much denser than SOTA methods. 2) To some extent, building SSGs can also be viewed as splitting the space into cones like RNG*. However, SSG indexing algorithm involves no randomness like RNG*. 3) Part of proofs of SSG's theoretical properties are inspired by RNG*(S) and MRNG. But SSGs are different from MRNG, NSG, and RNG*(S), because the pruning strategy of SSGs consider both length of edges and angle between edges, while those of the other methods only consider the lengths. See the Appendix for more details.


\begin{figure}[t]
\begin{center}
\includegraphics[width=170pt, height=120pt]{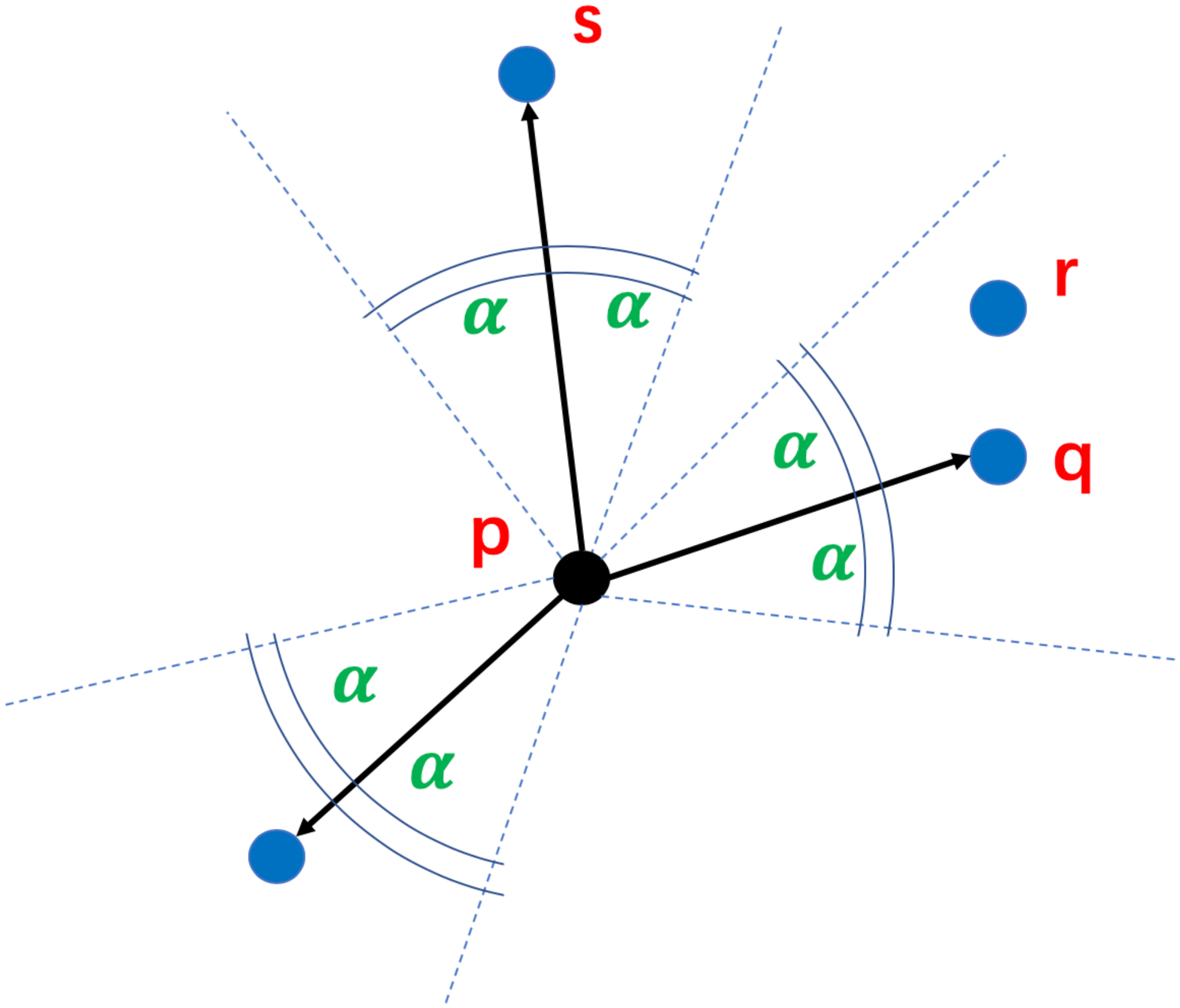}
\end{center}
   \caption{\textbf{An illustration of selecting edges for one point in a toy 2D SSG.} .}
\label{ssg_rule}
\end{figure}

\section{Methodology}
\label{method}
\subsection{Motivation}
As discussed in Sec. 2, many prior graph-based indices can be regarded as pruned graphs from the complete graph or a KNNG. They have basically reached a consensus that sparsifying the graph and reducing number of hops during search (with Alg. \ref{search_alg}) are both crucial to the search performance\cite{MalkovYHNSW16, Ben2016Fanng, fu2019fast}. However, they generally focus on accelerating the fast routing but neglect the importance of the sparsity. Interestingly, we find that there may exist an optimal degree for a given setting (data volume, distribution, etc). Sparser or denser graphs show inferior search performance (Fig. \ref{msnetangle}).

Specifically, typical SOTA algorithm such as NSG\cite{fu2019fast} and HNSW\cite{MalkovYHNSW16} only use ``distance-based'' pruning criteria (e.g., remove relatively longer edges as long as the connectivity is guaranteed). They do not discuss the influence of such criteria on the sparsity. In typical man-made message transfer system like communication satellites, both ``angle'' and ``distance'' are considered intuitively, i.e., neighbors are distributed omnidirectionally around each satellite, and communications mostly happen between nearest neighbors.

In addition, in high-dimensional similarity retrieval literature, indexed queries and unindexed queries are seldom distinguished. This may be because most algorithms may not behave differently treating two types of queries. However, we find it apparently influences graph-based methods (see an example in Fig. \ref{in-out-db-query}). To the best of our knowledge, existing SOTA graph-based algorithms have not discussed this problem formally. 

Based on above observations, we are seeking a novel graph index structure which 1) ensures fast traversal on the graph, 2) enables flexible sparsity adjustment for search performance optimization, and 3) takes unindexed queries into consideration. Specifically, we try to achieve goal 1) by designing a new graph-based index under the MSNET framework, and achieve goal 2) and 3) by involving both ``angle'' and ``distance'' into the pruning strategy. We use angle between edges to control the sparsity and distance between neighbors to control the locality sensitiveness. We will formally introduce SSG as follows.


\subsection{Satellite System Graph}
\label{definition}
We first give a formal definition of MSNET as follows:
\begin{definition}[MSNET]
    $G$ is an MSNET if and only if $\forall$ point $q, s \in G$, there exists at least one "monotonic path" $\{s(n_0),n_1,n_2,...,n_L,q\}$, satisfying that $\forall 0\le l< L, \delta(n_l,q) > \delta(n_{l+1},q)$\cite{dearholt1988monotonic}.
\end{definition}
It has been proved that the search trajectory produced by Alg. \ref{search_alg} on an MSNET is guaranteed to be a monotonic path\cite{fu2019fast}. The length expectation of the path over all node-pairs is $O(n^{\frac{1}{d}}\log n /\triangle{r})$, where $\triangle{r}$ is almost a constant in the empirical evaluation. 

Based on aforementioned motivation, we design a new pruning strategy to obtain an SSG. For any $p \in S$, we list the remaining points in ascending order of $\delta(p, \cdot)$. For any two edges $\overset{\longrightarrow}{pq}$ and $\overset{\longrightarrow}{pr}$, if $\angle rpq < \alpha$ (a hyper-parameter), we consider it as a ``conflict'' and discard the \textbf{longer} edge. This process is repeated until all points are pruned (Alg. \ref{ssgbuild_alg}, Fig. \ref{ssg_rule}). Intuitively, the edges around each node are distributed evenly and maintain a near-minimal neighborhood coverage like satellites. We formally define SSG as follows:
\begin{definition}[SSG]
In $E$, an $SSG$ is defined as the set of directed edges satisfying the property: for any edge $\overset{\longrightarrow}{pq}$, $\overset{\longrightarrow}{pq} \in SSG$ if and only if $Cone(\overset{\longrightarrow}{pq},\alpha) \cap B(p,\delta(p,q)) \cap S = \emptyset$ or $\forall r \in Cone(\overset{\longrightarrow}{pq},\alpha) \cap B(p,\delta(p,q)) \cap S, \overset{\longrightarrow}{pr} \notin SSG$, where $0 \le \alpha \le 60^{\circ}$ is a hyper-parameter.
\end{definition}

\smallskip
\noindent
\textbf{Properties of SSG} is as follows:

\begin{theorem}
\label{theorem1}
An SSG is an MSNET.
\end{theorem}

\begin{theorem}
\label{theorem2}
For $S$, randomly distributed in $E^d$, and any query $q \in S$, the search complexity from a random starting point is $O(Dn^{1/d}\log(n^{1/d})/\triangle{r})$, where $D$ is the degree upper-bound of SSG, $n$ is the size of the dataset, and $\triangle{r}$ is the lower-bound of length differences of edges in any non-isosceles triangles.
\end{theorem}

\begin{theorem}
\label{theorem3}
For $S$ (randomly distributed in $E^d$), any query $q \notin S$, and $q$'s nearest neighbor $r \in S$ , the probability that each step on the search path is monotonic to both $r$ and $q$ is $0.5+\epsilon, 0<\epsilon\le 0.5$, under the condition that for any neighbor $s$ of the node $p$ in the current search step, $\delta(p,s) < \delta(p,q)$. Further, $\epsilon = 0.5$ when we set $\alpha \le 30^{\circ}$ for the SSG. 
\end{theorem} Please refer to the proofs in the Appendix.

A simple explanation of above theorems and proofs is that, under a constraint on $\alpha$ ($0 \le \alpha \le 60^{\circ}$), SSG supports monotonic search for indexed queries. With a more strict constraint ($0 \le \alpha \le 30^{\circ}$), SSG guarantees monotonic search for both indexed and unindexed queries. In other words, the near-logarithmic theoretical search complexity can also apply to unindexed queries when $0 \le \alpha \le 30^{\circ}$. In addition, it is important to note that when $30 < \alpha \le 60^{\circ}$, searching for the unindexed queries is not guaranteed to be monotonic to the answers. This does not mean the search will "turn back". Because the search algorithm is greedy, it would just move towards a sub-optimal direction. In the worst case ($\alpha = 60^{\circ}$), there is still a probability higher than 0.5 to make the optimal choice at each time. As long as the optimal path is short, the sub-optimal path will not be long in expectation. Note that there is a pre-condition in Theorem \ref{theorem3}. The condition means that we only discuss all points on the search path except for the last a few steps. Consider that $\exists s$, a neighbor of $p$, satisfies that $\delta(p,s) \ge \delta(p,q)$. It means $q$ is already in $p$'s very close neighborhood. Because we only know the position of the unindexed query, but not the answer position, we use this as a terminal sign for generality.

Unlike previous works, SSG is not unique on a given dataset, a different $\alpha$ determines a different SSG. As discussed above, there is a trade-off between the graph sparsity and search path length. With a proper degree, we may achieve the optimal search performance. It is not difficult to find out that with a larger $\alpha$, the SSG will be more sparse. In theory, we can assign a different $alpha$ for each node adaptively to get an optimal SSG on a given dataset.

Note that above theoretical properties only hold for 1-NN search. Like most prior works \cite{arya1993approximate, fu2019fast}, we do not make more efforts on generalizing to the k-NN scenario. However, in most prior empirical evaluations, the k-NN performance can be approximately generalized from 1-NN.

\section{Navigating Satellite System Graph}
\label{NSSGsec}

\begin{table}[t]
\caption{The results of the SSGs and the MRNG on SIFT10K. AOD denotes the average out-degree of the graph, MOD denotes the max-out-degree of the graph, $L_{indexed}$ denotes the average search path length for indexed queries, and $L_{unindexed}$ denotes the average search path length for unindexed queries. SSG$_a^{\circ}$ denotes the SSG with $\alpha = a$, while SSG$_{a^{\circ}}tr$ denotes the truncated SSG$_a^{\circ}$. } 
\label{exp:exact}
\centering
\begin{tabular}{ccccc}
\hline
Graph & AOD & MOD & $L_{indexed}$ & $L_{unindexed}$\\
\hline
MRNG & 18 & 66 & 2.76 & 4.98 \\
SSG$_{60^{\circ}}$ & 40 & 111 & 2.18 & 3.91 \\
SSG$_{30^{\circ}}$  & 121 & 746 & 1.80 & 1.79 \\
SSG$_{60^{\circ}tr}$ & 19 & 40 & 2.20 & 3.95 \\
SSG$_{30^{\circ}tr}$  & 41 & 120 & 1.89 & 1.95 \\
\hline
\end{tabular}
\end{table}

\subsection{From SSG to NSSG}
SSG presents excellent ANNS properties theoretically, but there are mainly two problems with SSG: high indexing complexity (O$(n^3)$ for Alg. \ref{ssgbuild_alg}) and high degree of the graph. For practical use, we aim to propose a practical variant (named as Navigating SSG) which is both temporal-efficient and space-efficient. Naturally, \textbf{the theoretical properties of SSG will not hold for this variant,} but with proper heuristics, the actual performance is close to theoretical analysis. The heuristics are based on the observations in our empirical study as follows:

We design a small pre-experiment on the sift10K\footnote{http://corpus-texmex.irisa.fr/} dataset (10k points, 128 dimension) to show the relationship between graph degree and search path length. Different MSNETs (MRNG and different SSGs) are included in the comparison. MRNG is a theoretical graph model proposed in the NSG work \cite{fu2019fast}. The results are shown in Table \ref{exp:exact}. 

We can see that though the average search path length on the SSG is smaller, but the average degree of the graph is too large. As a result, searching on the exact SSGs has no advantages over the MRNG. Further, we perform this experiments on "truncated" SSGs. The truncated SSG is obtained by pruning the edges from the exact SSG. For example, in this experiment, we only keep 40 nearest neighbors and remove longer edges for each node in the SSG$_{60^{\circ}}$ to get SSG$_{60^{\circ}}tr$. We find that it almost makes no difference to the search path lengths, but the sparsity of the graph is increased significantly. In other words, this illustrates that most of the \textbf{long edges in the exact SSGs contribute little to the search routing}. The "effective" edges are mainly distributed in a small neighborhood around each node. 


\begin{algorithm}[t]
	\caption{NSSGIndexing($D$, $l$, $r$, $s$, $\alpha$)}
	\label{SSG_build_alg1}
	\begin{algorithmic}[1]
		\Require dataset $D$, candidate set size $l$, maximum out-degree $r$, number of navigating points $s$, minimum angle $\alpha$.
		\Ensure an NSSG $G$.
		\State Build an approximate $k$NN graph $G_{knn}$.
		\State $G=\emptyset$.
		\ForAll{node $i$ in $G_{knn}$}
		\State $P=\emptyset$.
		\ForAll{neighbor $n$ of node $i$}
		\State$P$.add($n$).
		\ForAll{neighbor $n'$ of node $n$}
		\State$P$.add($n'$).
		\EndFor
		\State remove the duplicated nodes in $P$.
		\If{$P$.size() $\ge l$}
		\State break.
		\EndIf
		\EndFor
		\State Perform SSG's pruning strategy on $P$.
		\State Update $G$ with selected edges. 
		\EndFor
		\ForAll{node $i$ in $G$}
		\ForAll{node $j$ in node $i$'s neighbors} 
		\State try adding node $i$ to $j$'s neighbors according
		\State to SSG's pruning criteria and avoid duplicates.
		\State remove longer edges if exceeding max-degree $r$.
		\EndFor
		\EndFor
		\State Random select $s$ points from the datasets as $NV$.
		\ForAll{point $i$ in $NV$}
		\State Strengthen the connectivity of the graph with 
		\State DFS-spanning from $i$.
		\EndFor
		\State return $G$.
	\end{algorithmic}
\end{algorithm}

If long edges are ``useless'' in SSG, we do not need to access them in the first place. Based on this observation, we can build a graph similar to SSG with low indexing complexity with the following steps: for each node, 1) generate a candidate neighbor set effectively covering the close neighborhood; 2) select the neighbors within this set through the same strategy in Alg. \ref{ssgbuild_alg}.

Step 1) can be achieved efficiently by building a high-quality approximate K-NN graph, where we can easily get the nearest k neighbors for each node. For most existing approximate K-NN graph construction algorithms, it is faster if $k(\ll n)$ is small. For a K-NN graph with a small $k$, we can expand the neighborhood coverage to get the candidate neighbor set by checking the neighbors' neighbors (2-hop neighbors). A sufficient neighborhood coverage is essential because the SSG edge selection process will refuse most nodes in the candidate set, and after the selection, we need ensure the out-edges evenly distributed in different directions for efficient search routing.

Step 2) is almost the same as the selection algorithm in Alg. \ref{ssgbuild_alg}, except that in Alg. \ref{ssgbuild_alg} the selection is performed on the edges between each node and all the rest $n-1$ nodes. We only perform the selection on the candidate set from Step 1). 

Considering that the "neighbor-propagation" in Step 1) may not be the most effective coverage of the neighborhood, we propose \textit{reverse neighbor checking} to refine the graph which is similar to the technique used in \cite{malkov2014approximate, MalkovYHNSW16, li2016approximate}. The motivation is that, for two nodes $p$ and $q$, if edge $\overset{\longrightarrow}{pq}$ is an effective in $p$'s neighborhood, it is possible that $\overset{\longrightarrow}{qp}$ is also effective for node $q$. We will try to insert $\overset{\longrightarrow}{qp}$ into $q$'s neighbor set. If $\overset{\longrightarrow}{qp}$ conflicts with another out-edge of $q$, it means there is already an effective edge in this direction given $\alpha$. 

Further, we adopt the technique in \cite{fu2019fast} to ensure the connectivity of the graph. Specifically, we select a few navigating nodes (randomly) and ensure the connectivity from these nodes to all the others. Different from the standard SSG, NSSG is not guaranteed to be strongly connected. Meanwhile, it is time consuming to turn a general graph into a strongly connected one. Instead, we randomly select a few navigating nodes from the dataset, treat them as roots, and span DFS trees from them. Please refer to paper \cite{fu2019fast} for more details. 

Please refer to Alg. \ref{SSG_build_alg1} for details on NSSG indexing. The search algorithm on the NSSG is the same as Alg. \ref{search_alg}, except that searching on an NSSG starts from the nearest navigating node to the query. In this way, we can get more stable search performance.

\subsection{Analysis}
\smallskip
\noindent
\textbf{Complexity.} The indexing of NSSG contains mainly two stages: 1) K-NN graph construction and 2) edge processing. The K-NN graph construction is not in the scope of this paper. Given a K-NN graph, the edge processing, including candidate edge generation, edge selection and reverse neighbor checking, can be done within linear time according to Alg. \ref{SSG_build_alg1}. The edge generation is simply two rounds of "neighbor-fetch" on the pre-built approximate K-NN graph. The complexity is O($nk^2$), where $k (\ll n)$ is the degree of the K-NN graph. The edge selection includes $n(k^2+rk^2)$ times of distance calculation and a sorting on a set of $k^2$ node, where $r (\ll n)$ is the max-degree constraint of the NSSG. The complexity is O$(dn(k^2+rk^2) + k^2\log k^2)$. In the worst case, i.e., the reverse neighbor checking replaces all the nodes in the NSSG, the complexity of reverse neighbor checking is the same as the edge selection. In summary, the total complexity of the second stage, edge processing, is O$(dn(k^2+rk^2)+nk^2 + k^2\log k^2)$. The DFS-spanning can also be done in linear time \cite{fu2019fast}, which is O$(nrv)$, where $v$ is the number of navigating nodes. Consequently, the edge processing stage has linear complexity regarding the dataset size $n$. However, the K-NN graph construction usually has a complexity much higher than linear time \cite{Dong2011Efficient}, which dominates the indexing of NSSG. 

In our later empirical evaluation, the search complexity on the NSSG is about O($n^{\frac{1}{d}} \log n$), close to the theoretical complexity of the SSG, which indicates that our NSSG indexing algorithm can produce a graph best approximating the properties of SSG.

\smallskip
\noindent
\textbf{Worst case.} Given enough iterations, the answer is guaranteed to be found because the connectivity from the navigating node to all the other nodes is guaranteed. The worst case is checking all the nodes in the graph, which can hardly happen. In that case, all the nodes should almost form a straight line in the space, and the navigating node is on one end of the line while the query lies on the other end. Such extreme case will hardly happen in practical problems.

\section{Experiments and Analysis}
\label{experiment}
In this section, we will present the experimental results and analysis to demonstrate the effectiveness of our method. 

\subsection{Datasets and Experimental Setting}
We use five different public real-world datasets to evaluate the performance of different algorithms, including four million-scale datasets, \eg, SIFT1M\footnote{http://corpus-texmex.irisa.fr/}, GIST1M\footnote{http://corpus-texmex.irisa.fr/}, Crawl\footnote{http://commoncrawl.org/}, and GloVe\footnote{https://nlp.stanford.edu/projects/glove/}, and one large-scale dataset, Deep100M, with 100 million points. The text (Crawl) is converted to vectors with fastText\footnote{https://fasttext.cc/ use mean pooling to get doc-embedding}. The Deep100M is sampled from the Deep1B dataset\footnote{https://github.com/facebookresearch/faiss/tree/master/ benchs\#getting-deep1b}. The datasets come from different media (\ie, texts or images) and different feature descriptors, thus their data distribution and the intrinsic dimensions are quite different. In this way, the results can reflects the generalizability of the algorithms on different data distribution. Please refer to Table \ref{dataset} for more details.

\smallskip
\noindent
\textbf{Machine Configuration.} All the are conducted on a machine with an Intel i9-7980XE CPU, 128 GB memory, 2 TB disk space and 4 NVIDIA 1080Ti GPUs. The indices are built with 36 threads in parallel. On the four smaller million-scale datasets, the search is evaluated in a single-threaded manner for fair comparison. On Deep100M, only NSG, NSSG, and Faiss are included in the evaluation because the other algorithms will trigger the out-of-memory (OOM) error on this machine. Only the Annoy algorithm uses the HDD-disk to store the index at runtime. All the other algorithms load their index into memory.

\subsection{Evaluation.} 
We evaluate the search performance of the compared algorithms with the criteria --- queries-per-second v.s. precision, \ie, we record how much queries an algorithm can process per second at given precision. The precision is formally defined as:
\begin{equation}
    precision(R) = \frac{|R \cap G|}{|G|}
\end{equation}
where $R$ is the answer set returned by the algorithm, and $G$ is the ground-truth set of the given query \cite{fu2019fast}.

Real applications are usually more sensitive to the search time in high-precision area. With different hyper-parameters, each algorithm will perform very differently in different precision area, we will present the best performance of all the algorithms in the high-precision region via parameter-tuning. In order not to lose generality, we divide each base set into training set, and validation set by 99:1. We build the indices of different algorithms on the training set and tune the parameters on the validation set to get the best-performing index of each algorithm at high-precision. 

We also evaluate the indexing process of graph-based methods in terms of indexing time, memory usage, and the sparsity of the graphs.

\subsection{Compared Algorithms.} 
We select seven state-of-the-art algorithms from different types. We do not compare with the hashing-based methods because they are generally too slow at high-precision region (much slower than the serial-scan). 

We \textbf{do not include SSG} because 1) SSG's indexing is too slow to do parameter tuning ($\alpha$) at million scale; 2) the small-scale experiment in Sec. \ref{NSSGsec} has shown that SSG suffers from less contributive long edges and extremely high degree. It is foreseeable that SSG would perform much worse than NSSG in large-scale. The performance of SSG will be much worse Specifically, the selected algorithms are listed as follows:
\begin{table}[t]
\caption{Information of experimental datasets. We list the dimension (D), the local intrinsic dimension (LID), the number of base vectors, and the number of query vectors.}
\label{dataset}
\centering
\begin{tabular}{ccccc}
\hline
Dataset & D & LID & No. of base & No. of query\\
\hline
SIFT1M & 128 & 12.9 & 1,000,000 & 10,000 \\
GIST1M & 960 & 29.1 & 1,000,000 & 1,000 \\
Crawl  & 300 & 15.7 & 1,989,995 & 10,000 \\
GloVe  & 100 & 20.9 & 1,183,514 & 10,000 \\
Deep100M & 96 & 10.2 & 100,000,000  & 10,000 \\
\hline
\end{tabular}
\end{table}

\begin{figure*}[t]
	\centering
	\subfigure{\includegraphics[width=1\textwidth]{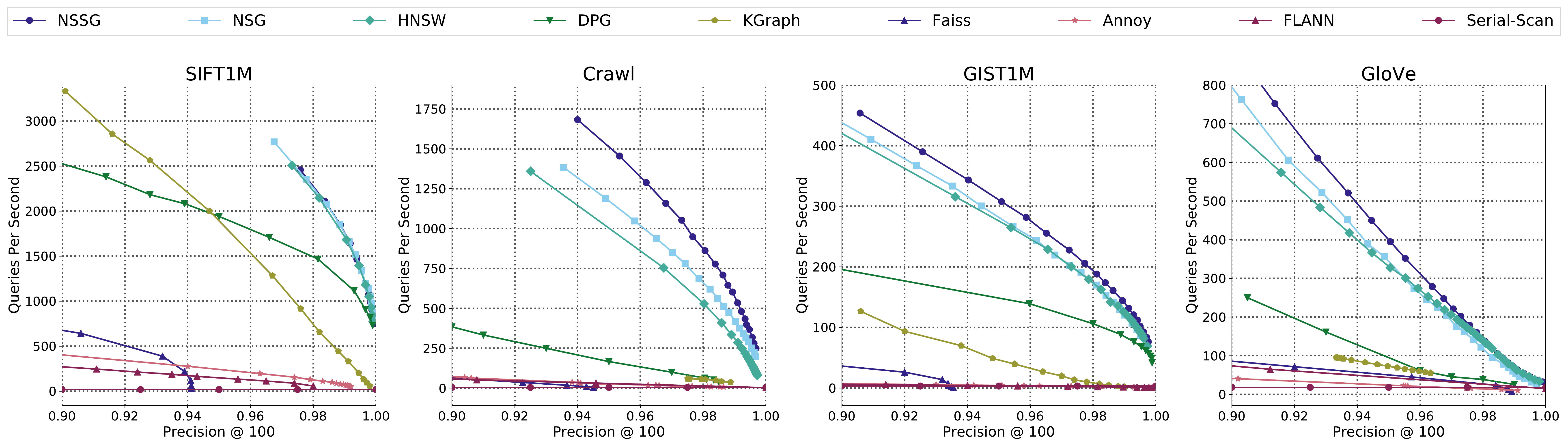}}
	\caption{\textbf{ANNS performance of graph-based algorithms (in log-scale) with their optimal indices in high-precision region on the four datasets} (top right is better). The x-axis is not meaningful for Serial-Scan because the results are accurate. }
	\label{graph_search}
\end{figure*}

\begin{table}[t]
\caption{Information of the graph-based indices involved in all of our experiments. The Size means the memory occupation of the index. AOD means the Average Out-Degree. MOD means the Maximum Out-Degree. Because HNSW contains multiple graphs, we only report the AOD and MOD of its bottom-layer graph here.}
\label{best_index_tb}
\centering
\begin{tabular}{ccccc}
    \hline
    Dataset & Algorithms & Size (MB) &  AOD & MOD  \\
    \hline
    
    \multirow{5}{*}{SIFT1M} 
    & NSSG & \textbf{153} &  39 & 50  \\
    & NSG & \textbf{153} & 25.9 & 50  \\
    & HNSW & 451 & 32.1 & 50\\
    & KGraph & 374  & 200 & 200 \\
    & DPG & 632  & 165.1 & 1260 \\
    \hline
    
    \multirow{5}{*}{GIST1M} 
    & NSSG & \textbf{267}  & 34 & 70  \\
    & NSG & \textbf{267}  & 26.3 & 70  \\
    & HNSW & 667  & 23.9 & 70 \\
    & KGraph & 1526 & 400 & 400 \\
    & DPG & 741  & 194.3 & 20899 \\
    \hline
    
    \multirow{5}{*}{Crawl} 
    & NSSG & \textbf{303}  & 22 & 40  \\
    & NSG & \textbf{303}  & 11 & 40  \\
    & HNSW & 759 &  12.1 & 40 \\
    & KGraph & 3036  & 400 & 400 \\
    & DPG & 1465  & 193 & 97189 \\
    \hline
    
    \multirow{5}{*}{GloVe} 
    & NSSG & \textbf{225}  & 21 & 50  \\
    & NSG & \textbf{225}  & 13 & 50  \\
    & HNSW & 564  & 12 & 50 \\
    & KGraph & 1805  & 400 & 400 \\
    & DPG & 787  & 174 & 50336 \\
    \hline
\end{tabular}
\end{table}

\begin{enumerate}
    \item \textbf{FLANN\footnote{https://github.com/mariusmuja/flann}} is a well-known ANNS library based on many tree-based algorithms, including randomized KD-trees, Kmeans trees, and so on. We use its auto-tune composite tree algorithm for comparison. 
    \item \textbf{Annoy\footnote{https://github.com/spotify/annoy}} is a K-means tree algorithm with $K=2$. The algorithm is specially optimized for $K=2$.
    \item \textbf{Faiss\footnote{https://github.com/facebookresearch/faiss}} is a quantization-based algorithm recently released by Facebook. We use its IVF-PQ implementation for comparison. Specifically, the index contains two parts: the inverted file (IVF) and the product quantization code and the codebook. The search on the IVFPQ is a two-stage process: use the IVF to locate a small number of candidate answers and then use the quantized distance to rank them.
    \item \textbf{KGraph\footnote{https://github.com/aaalgo/kgraph}} is a graph-based method which use a KNN graph as the index. It is a well-implemented version of the GNNS \cite{Hajebi2011Fast} algorithm. It also contains a well-implemented $nn$-descent algorithm \cite{Dong2011Efficient} to build the KNN graph. 
    \item \textbf{DPG\footnote{https://github.com/DBWangGroupUNSW/nns\_benchmark}} is a graph-based method which also uses the angles between the edges to select edges \cite{li2016approximate}. The differences between the DPG and the SSG are discussed in Sec. \ref{related_work}.
    \item \textbf{HNSW\footnote{https://github.com/searchivarius/nmslib}} is a well-known graph-based algorithm based on Hierarchical NSW Graph \cite{MalkovYHNSW16}.
    \item \textbf{NSG\footnote{https://github.com/ZJULearning/nsg}} is the state-of-the-art graph-based algorithm which is the approximation of the MRNG \cite{fu2019fast}.
    \item \textbf{Serial-Scan}. We perform serial-scan on the test sets to show the speed-up of different algorithms. For clearness, we draw a horizontal line for Serial-Scan instead of a point to show its positions in the figures.
    \item \textbf{NSSG} uses Alg. \ref{SSG_build_alg1} to build the NSSG. The KNN graph construction algorithm we use here is the $nn$-descent algorithm \cite{Dong2011Efficient}.
\end{enumerate}

\section{Results and Analysis}

\subsection{Search Performance.}

The search performance on the four medium-scale datasets are shown in Fig. \ref{graph_search}. We can see that:
\begin{enumerate}
    \item Our approach outperforms the others significantly. Especially on datasets with higher intrinsic dimension, the gap between our method and the others becomes wider and wider. On Crawl, GloVe, and GIST1M, NSSG achieves 35\%, 15\%, and 20\% speed-up over NSG respectively at 95\% precision. This indicates our method is insensitive to the increase of the dimension compared with the others.
    
    \item Graph-based methods are much better than non-graph-based ones on these datasets in high-precision, which is consistent with the results in the experimental report of \cite{Hajebi2011Fast, Ben2016Fanng, fu2019fast, MalkovYHNSW16}.
    
    \item NSSG is neither the most sparse nor the most dense graph among all graph indices, but NSSG performs the best. This agrees with our hypothesis that there exists an optimal degree. With different $\alpha$, we can easily adjust NSSG's sparsity and meanwhile preserve the excellent theoretical properties. In contrast, NSG and HNSW are too sparse, thus their performance is inferior.
    
    \item DPG performs much worse than the NSSG. The main reason is that DPG is too dense (Table \ref{best_index_tb}) and no theory supports DPG's edge selection strategy. 
    
    \item On SIFT1M, NSSG, NSG, and HNSW show almost the same search performance. This is because the SIFT1M is of the lowest intrinsic degree among all the four datasets, in other words, the SIFT1M is the simplest dataset. NSSG, NSG, and HNSW achieve the performance up-limit of graph-based methods.
    
    \item The curves of NSSG, NSG, and HNSW have relatively high-accuracy starting points, because their worst search parameters produce rather high accuracy. In terms of usability, HNSW is easier to use because it does not need to build a KNN graph.
    
    \item Faiss has an obvious accuracy up-limit on all datasets. This is because the quantization error limits its accuracy upper-bound. No matter how we tune the parameters, we cannot get a better result. With higher dimension, the quantization error grows quickly. NSSG does not have such upper-bound because we search based on accurate distance calculation rather than quantized distance.

\end{enumerate}

\begin{figure*}[t]
\begin{center}
\includegraphics[width=1\textwidth]{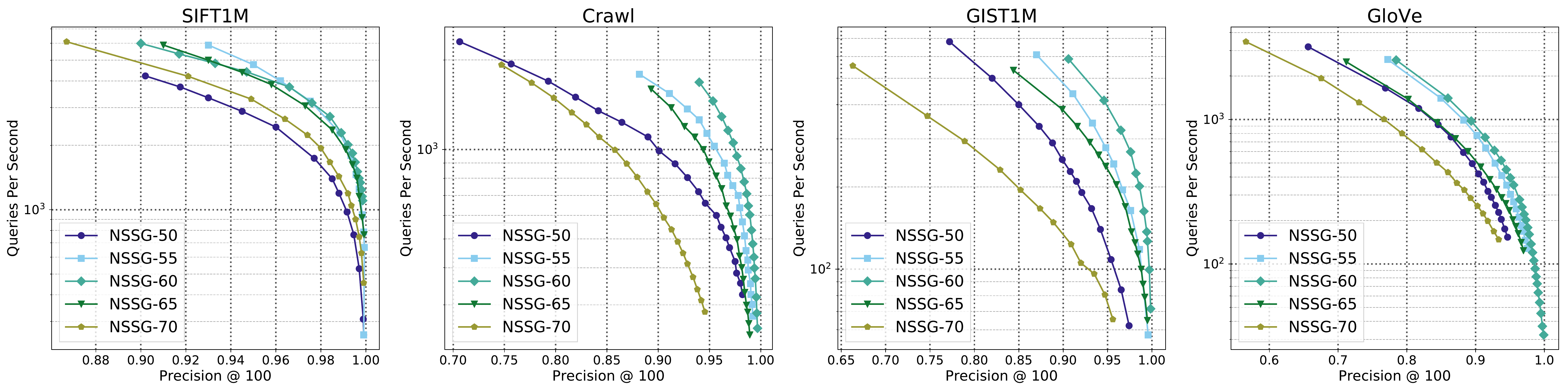}
\end{center}
   \caption{\textbf{The results of using different angles on GIST1M datasets.} On SIFT1M,  $55^{\circ}$ is the best performing angle, but in general, the performance with $60^{\circ}$ is  better than the other angles.}
\label{angle}
\end{figure*}

\begin{table}[t]

	\caption{The indexing time of all the graph-based methods. The indexing times of NSSG and NSG are recorded in the form $t_1+t_2$, where $t_1$ is the time to build the approximate KNN graph, and $t_2$ is the time of edge-selection and connectivity strengthening. }
	\label{index_time}
	\centering
	\scalebox{0.85}{
	\begin{tabular}{cccccc}
		\hline
		Dataset & NSSG  & NSG & HNSW & KGraph & DPG\\
		\hline
		SIFT1M
		&\textbf{62+13 (s)}  & 62+45 (s) & 149 (s) & \textbf{62 (s)} &1120 (s)\\
	
	    GIST1M
		&\textbf{620+144 (s)} & 620+735 (s) & 1376 (s) &\textbf{620 (s)}&6700 (s) \\
		
		Crawl
		&\textbf{790+82 (s)} & 790+567 (s) & 1083 (s) &\textbf{790 (s)} & 9169 (s) \\
		
		GloVe
		&\textbf{650+18 (s)} & 650+102 (s) & 930 (s) &\textbf{650 (s)} & 3139 (s) \\
		\hline
	\end{tabular}
	}
\end{table}

\subsection{Indexing Performance.} 
We lists the details of indices of various algorithms in Table \ref{best_index_tb} and \ref{index_time}. The index size of the graph is determined by the max-degree of the graph, because almost all the graph-based methods store their graphs as a $n\times m$ matrix, where $n$ is the number of nodes and $m$ is the max-degree. In this way, we can easily access the neighbors of each node within a continuous block of memory. This operation is not suitable for the DPG because its max-out-degree is too large. DPG uses a 2D array to store the graph. We can see that:
\begin{enumerate}
    \item NSSG and NSG's indices are the smallest among all the graph-based methods because their optimal indices have a small max-degree. Though the max-degree of the HNSW is the same with the NSSG and the NSG, the index size of the HNSW is large because it contains multiple layers of graphs.
    \item The approximate K-NN graph construction time dominates the total indexing time of NSSG, which agrees with our analysis.
    \item The edge selection ($t_2$) of NSSG is 3-7 times faster than NSG, because the time complexity of NSSG 's edge selection is much lower than that of NSG's.
    \item For fair comparison, NSSG uses the same K-NN graph construction algorithm as NSG does (the nn-descent algorithm, i.e., KGraph). This process can be accelerated tens of times by replacing nn-descent by other methods, such as the GPU-version of Faiss.
    \item The indexing time of non-graph-based methods are not listed here. It is necessary to mention that the indexing of graph-based methods is usually slower than non-graph-based ones.     
\end{enumerate}

\subsection{Parameters}
There are several parameters influencing the performance of NSSG: the size of the candidate-neighbor-pool $l$, the max-out-degree of the graph $r$, the number of navigating nodes $m$, the minimal angle between edges $\alpha$.

We test multiple different angles on different datasets (Fig. \ref{angle}). On SIFT1M, NSSG-55 is the best but the gap between NSSG-60 and NSSG-55 is not obvious. On the rest dataset, NSSG-60 is the best. Thus, we recommend $\alpha=60^{\circ}$ for new users. NSSGs with smaller $\alpha$ performs worse because the graph becomes more dense. The search paths are shortened, but the degree of the graph increases too much. NSSGs with larger $\alpha$ performs worse because it no longer preserves the MSNET property and is too sparse. 

For the remaining indexing parameters of all the compared algorithms, we perform grid-search on the validation set to obtain the optimal values of important parameters of different algorithms on different datasets as follows:

    1) \textbf{SIFT1M.} We use $l=100,r=50,\alpha=60^{\circ}, m=10$ for NSSG, $l=40,r=50$ for NSG, $M=25,efconstruction=600$ for HNSW, $K=400,L=200,S=10,R=100,I=20$ for KGraph, $L_1=500, S=20, K=400, L_2=200$ for DPG, $nlists=4096$ for Faiss, $n_trees=400$ for Annoy and \textit{autotuned} algorithm for FLANN.
    
    2) \textbf{GIST1M.} We use $l=500,r=70,\alpha=60^{\circ}, m=10$ for NSSG, $l=60,r=70$ for NSG, $M=35,efconstruction=800$ for HNSW, $K=400,L=400,R=100,S=15,I=12$ for KGraph, $L_1=500, S=20, K=400, L_2=200$ for DPG, $nlists=4096$ for Faiss, $n_trees=400$ for Annoy and \textit{autotuned} algorithm for FLANN.
    
    3) \textbf{Crawl.} We use $l=500,r=40,\alpha=60^{\circ}, m=10$ for NSSG, $l=500,r=40$ for NSG, $M=20,efconstruction=1000$ for HNSW, $K=400,L=420,S=15,R=100,I=12$ for KGraph, $L_1=500, S=20, K=400, L_2=200$ for DPG, $nlists=4096$ for Faiss, $n_trees=400$ for Annoy and \textit{autotuned} algorithm for FLANN.
    
    4) \textbf{GloVe.} We use $l=500,r=50,\alpha=60^{\circ}, m=10$ for NSSG, $l=150,r=50$ for NSG, $M=25,efconstruction=2500$ for HNSW, $K=400,R=420,S=20,R=200,I=12$ for KGraph, $L_1=500, S=20, K=400, L_2=200$ for DPG, $nlists=4096$ for Faiss, $n_trees=500$ for Annoy and \textit{autotuned} algorithm for FLANN.

As for the meanings of the parameter symbols, please refer to the corresponding GitHub pages. 

\subsection{Complexity}
\begin{figure*}[t]
\begin{center}
\includegraphics[width=1\textwidth]{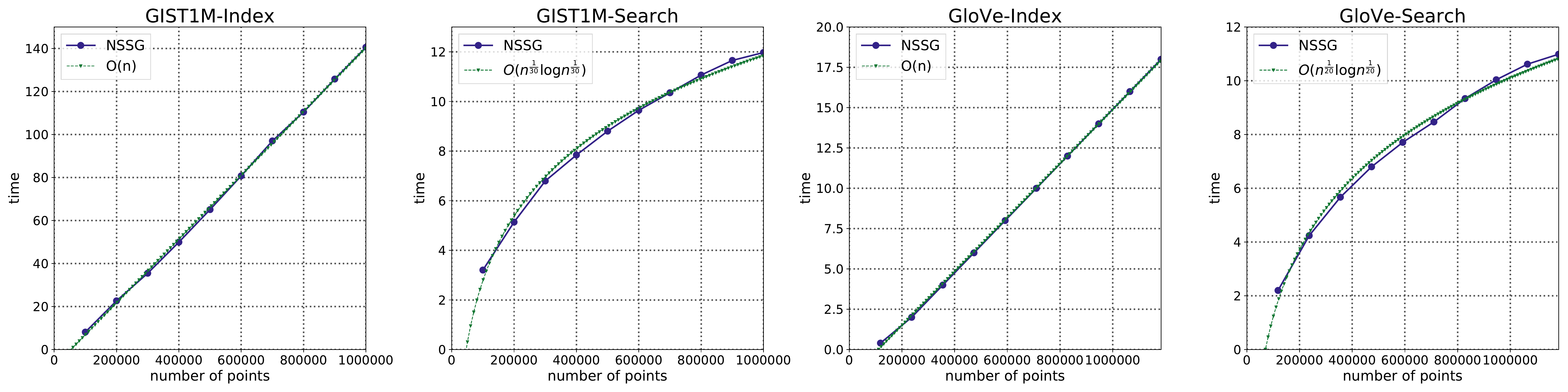}
\end{center}
   \caption{\textbf{The results on indexing complexity (edge selection) and searching complexity experiments on the GIST1M and GloVe dataset.} The solid lines are the complexity evaluation results. The dotted lines are the automatic fitting curves according to the formulation in the big O notation. Due to the randomness in the procedure, we evaluate the algorithm on different datasets multiple times to get the average.}
\label{complexity}
\end{figure*}

We split the medium-scale datasets into several subsets with different sizes, build the NSSG indices on these datasets, and record the searching and indexing performance statistics. To evaluate the indexing complexity, we use the same parameters to construct indices on all the subsets and record the time. It is important to note that the indexing time we report is only the time of the edge selection stage because K-NN graph construction is not in the scope of this paper.

To evaluate the search complexity, we conduct search with above indices on respective subset. Specifically, we adjust the search parameters until the result reaches 99\% precision. Due to the space limitation, we only show the search and indexing complexity experiments on the GIST1M and GloVe dataset in Fig. \ref{complexity}. The behaviors of the NSSG on other datasets follow similar patterns. 

In Fig. \ref{complexity}), the empirical estimation of the NSSG's edge selection complexity is O(n), while the estimation of the NSSG's search complexity is about $O(n^{\frac{1}{d}}\log n^{\frac{1}{d}})$, where $d$ approximately equals to the intrinsic dimension of the respective dataset. $O(n^{\frac{1}{d}}\log n^{\frac{1}{d}})$ is very close to O($\log n$) when $d$ is large. This agrees with our theoretical analysis.

We also analyze how much NSSG approximate SSG, please see the experiment result in the Appendix.

\subsection{Performance on Deep100M}
\label{deep1eExp}
\begin{figure}[t]
\begin{center}
\includegraphics[width=0.40\textwidth]{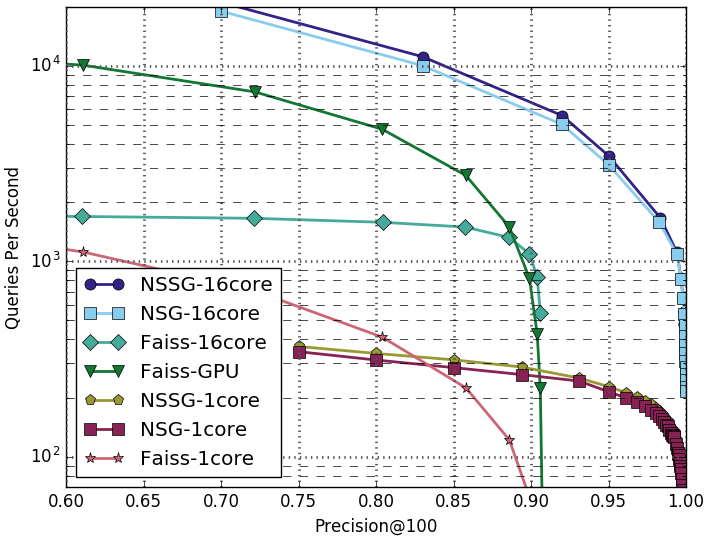}
\end{center}
   \caption{\textbf{The results on Deep100M in log-scale.} *-16core indicates we evaluate the search performance of the parallel-search version of the corresponding algorithms with 16 threads. *-1core indicates we evaluate the search performance of the sequential-search version of the corresponding algorithms with only one thread. We evaluate Faiss-GPU version on a NVIDIA 1080Ti GPU. We draw all these curves within one figure to show the performance difference clearly.}
\label{deep1e}
\end{figure}

We only evaluate the performance of the NSSG, NSG, and Faiss on the Deep100M datasets, because only these three algorithms can scale to such a big dataset. We evaluate these algorithms in two settings: \textbf{single-thread test and parallel-test}. Note that the parallel-test is not simply process different queries on different threads. It is meaningless regarding single-query-latency. Some algorithms can process one query in parallel, e.g., Faiss can do this by checking multiple inverted lists in parallel for one query. NSSG can do this via partitioning the dataset evenly into several chunks, building multiple sub-graphs subsequently, and searching them in parallel.

We use Faiss to build the KNN graph for NSSG on the Deep100M dataset, because $nn$-descent cannot scale to such a large dataset. Similar to the evaluation setting in \cite{fu2019fast}, we first build one NSSG index on the whole Deep100M to evaluate NSSG performance in the single-thread setting (NSSG-1core in Fig. \ref{deep1e}). Then Deep100M dataset is randomly divided into 16 subsets, 6.25M points in each. Then we build 16 NSSG graphs on these 16 subsets respectively. When we perform k-NN search for a query, we search all the 16 indices in parallel and aggregate the results to get the final answers. 

In our experiments, we evaluate the 100-NN performance for each algorithms in both single-thread test and parallel test settings. The results are shown in Fig. \ref{deep1e}. Suffix ``16core'' denotes the search performance of a given algorithm in parallel-test, while ``1core'' denotes the corresponding search performance in the single-thread test setting. Further, we also draw the Faiss search performance on a GPU as reference.

We can see that the NSSG outperforms the NSG and Faiss on both single-thread and parallel settings. NSSG achieves 5\% to 15\% times speed-up over the NSG under different accuracy, and hundreds times speed-up over Faiss.

In terms of indexing, NSSG achieves 30\% speed-up over NSG. Specifically, in single-thread setting, it takes 11 hours to build the NSSG-1core index (single-graph index), while 16 hours for the NSG-1core index. In parallel setting, NSG uses 3.53 hours to build the NSG-16core index, whereas we spend only 2.6 hours to build the NSSG-16core index (multi-graph index). This result again shows that the indexing complexity of the NSSG is much lower than the NSG's. 

In addition, Alg.~\ref{SSG_build_alg1} can be further optimized for large datasets. Specifically, we do not need to load a KNN graph into memory. We search for the K nearest neighbors for each node on-the-fly instead of saving the K-NN graph (please refer to the Alg. \ref{SSG_build_alg2} for more details of the implementation in this paper). Without the in-memory K-NN graph, we only use 56 GB memory at most, which is 37 GB less than the NSG, exactly the size of the K-NN graph. We can conclude that the NSSG is more scalable in large applications. 

\begin{algorithm}[t]
	\caption{NSSGIndexing($D$, $l$, $r$, $s$, $\alpha$)}
	\label{SSG_build_alg2}
	\begin{algorithmic}[1]
		\Require dataset $D$, candidate set size $l$, maximum out-degree $r$, number of navigating nodes $s$, minimum angle $\alpha$.
		\Ensure an NSSG $G$.
		\State Build an ANNS index $H$ with IVFPQ (Faiss) on $D$.
		\State $G=\emptyset$.
		\ForAll{node $i$ in $D$}
		\State $P= l$ nearest neighbors of node $i$ via search on $H$.
		\State Perform SSG pruning strategy on $P$.
		\State Update $G$ with selected edges. 
		\EndFor
		\ForAll{node $i$ in $G$}
		\ForAll{node $j$ in node $i$'s neighbors} 
		\State try adding node $i$ to $j$'s neighbors according
		\State to SSG's pruning criteria and avoid duplicates.
		\State remove longer edges if exceeding max-degree $r$.
		\EndFor
		\EndFor
		\State Randomly select $s$ nodes from the datasets as $NV$.
		\ForAll{node $i$ in $NV$}
		\State Strengthen the connectivity of the graph with 
		\State DFS-spanning from node $i$.
		\EndFor
		\State return $G$.
		
	\end{algorithmic}
\end{algorithm}

\section{Discussion and Limitations}

1) NSSG provides tens to hundreds times search speed-up over non-graph-based methods, and meanwhile the results are of very high accuracy. However, the indexing cost and memory usage of NSSG is usually higher than many non-graph-based methods. such as LSH and IVFPQ (Faiss). Therefore, NSSG is suitable for the scenarios which are insensitive to memory usage and data updating, but have high requirements for search latency and accuracy.

2) As shown in Sec. \ref{deep1eExp}, indexing a 100-million-point dataset on a single machine can be done within 3 hours, which is acceptable for daily-updating applications. Incremental indexing is feasible for NSSG, but it is not in the scope of this paper. We will leave it for future work.

3) In this paper, the biggest dataset we use includes 100 million points for a single machine, which is limited by the machine's physical capacity. To support bigger datasets with billions of points, the simplest approach is to use larger memory and better CPU. The other way is to implement the algorithm in an on-disk manner. Though we implement NSSG in an in-memory manner in this paper and GitHub, a simple on-disk alternative is to map the memory of the index and data to the disk with system calls interfaces. The search performance will decline due to the page fault latency. In our future work, we will explore more efficient on-disk algorithms.

\section{Conclusion}
In this paper, we propose a novel graph structure named as Satellite System Graph, which has superior theoretical guarantees for both indexed and unindexed queries. We can adjust the degree of SSG adaptively for different distribution. For large-scale applications, we propose the NSSG algorithm which has low indexing complexity and high search performance. Theoretical analysis and extensive experiments on several public datasets have demonstrated the strengths of our methods over the other SOTA approaches. Our code has been released on GitHub\footnote{https://github.com/ZJULearning/SSG}. 

\ifCLASSOPTIONcompsoc
  \section*{Acknowledgments}
\else
  \section*{Acknowledgment}
\fi

This work was supported in part by The National Key Research and Development Program of China (Grant Nos: 2018AAA0101400), in part by The National Nature Science Foundation of China (Grant Nos: 62036009, 61936006), in part by the Alibaba-Zhejiang University Joint Institute of Frontier Technologies.

\ifCLASSOPTIONcaptionsoff
  \newpage
\fi



%

\bibliographystyle{IEEEtran}
\bibliography{ssg}

\begin{thebibliography}{10}
\providecommand{\url}[1]{#1}
\csname url@samestyle\endcsname
\providecommand{\newblock}{\relax}
\providecommand{\bibinfo}[2]{#2}
\providecommand{\BIBentrySTDinterwordspacing}{\spaceskip=0pt\relax}
\providecommand{\BIBentryALTinterwordstretchfactor}{4}
\providecommand{\BIBentryALTinterwordspacing}{\spaceskip=\fontdimen2\font plus
\BIBentryALTinterwordstretchfactor\fontdimen3\font minus
  \fontdimen4\font\relax}
\providecommand{\BIBforeignlanguage}[2]{{%
\expandafter\ifx\csname l@#1\endcsname\relax
\typeout{** WARNING: IEEEtran.bst: No hyphenation pattern has been}%
\typeout{** loaded for the language `#1'. Using the pattern for}%
\typeout{** the default language instead.}%
\else
\language=\csname l@#1\endcsname
\fi
#2}}
\providecommand{\BIBdecl}{\relax}
\BIBdecl

\bibitem{BeisL97Shape}
J.~S. Beis and D.~G. Lowe, ``Shape indexing using approximate nearest-neighbour
  search in high-dimensional spaces,'' in \emph{1997 Conference on Computer
  Vision and Pattern Recognition}, 1997, pp. 1000--1006.

\bibitem{ferhatosmanoglu2001approximate}
H.~Ferhatosmanoglu, E.~Tuncel, D.~Agrawal, and A.~El~Abbadi, ``Approximate
  nearest neighbor searching in multimedia databases,'' in \emph{Data
  Engineering, 2001.}\hskip 1em plus 0.5em minus 0.4em\relax IEEE, 2001, pp.
  503--511.

\bibitem{chen2005robust}
L.~Chen, M.~T. {\"O}zsu, and V.~Oria, ``Robust and fast similarity search for
  moving object trajectories,'' in \emph{Proceedings of the 2005 ACM
  SIGMOD}.\hskip 1em plus 0.5em minus 0.4em\relax ACM, 2005, pp. 491--502.

\bibitem{philbin2007object}
J.~Philbin, O.~Chum, M.~Isard, J.~Sivic, and A.~Zisserman, ``Object retrieval
  with large vocabularies and fast spatial matching,'' in \emph{Computer Vision
  and Pattern Recognition, 2007. CVPR'07. IEEE Conference on}.\hskip 1em plus
  0.5em minus 0.4em\relax IEEE, 2007, pp. 1--8.

\bibitem{LiuRR07Clustering}
T.~Liu, C.~R. Rosenberg, and H.~A. Rowley, ``Clustering billions of images with
  large scale nearest neighbor search,'' in \emph{8th {IEEE} Workshop on
  Applications of Computer Vision}, 2007, p.~28.

\bibitem{zheng2016lazylsh}
Y.~Zheng, Q.~Guo, A.~K. Tung, and S.~Wu, ``Lazylsh: Approximate nearest
  neighbor search for multiple distance functions with a single index,''
  \emph{Proceedings of the 2016 International Conference on Management of
  Data}, pp. 2023--2037, 2016.

\bibitem{AroraSK018}
\BIBentryALTinterwordspacing
A.~Arora, S.~Sinha, P.~Kumar, and A.~Bhattacharya, ``Hd-index: Pushing the
  scalability-accuracy boundary for approximate knn search in high-dimensional
  spaces,'' \emph{{PVLDB}}, vol.~11, no.~8, pp. 906--919, 2018. [Online].
  Available: \url{http://www.vldb.org/pvldb/vol11/p906-arora.pdf}
\BIBentrySTDinterwordspacing

\bibitem{Bentley1975Multidimensional}
J.~L. Bentley, ``Multidimensional binary search trees used for associative
  searching,'' \emph{Communications of the {ACM}}, vol.~18, no.~9, pp.
  509--517, 1975.

\bibitem{Fukunaga1975A}
K.~Fukunaga and P.~M. Narendra, ``A branch and bound algorithm for computing
  k-nearest neighbors,'' \emph{IEEE Transactions on Computers}, vol. 100,
  no.~7, pp. 750--753, 1975.

\bibitem{Silpaanan2008Optimised}
C.~Silpa-Anan and R.~Hartley, ``Optimised kd-trees for fast image descriptor
  matching,'' in \emph{Proceedings of the 2008 IEEE Conference on Computer
  Vision and Pattern Recognition}, 2008, pp. 1--8.

\bibitem{Jagadish2005iDistance}
H.~V. Jagadish, B.~C. Ooi, K.~L. Tan, C.~Yu, and R.~Zhang, ``idistance: An
  adaptive b + -tree based indexing method for nearest neighbor search,''
  \emph{{ACM} Transactions on Database Systems}, vol.~30, no.~2, pp. 364--397,
  2005.

\bibitem{Fu2000Dynamic}
A.~W. Fu, P.~M. Chan, Y.~L. Cheung, and Y.~S. Moon, ``Dynamic vp-tree indexing
  for n-nearest neighbor search given pair-wise distances,'' \emph{{VLDB}
  Journal}, vol.~9, no.~2, pp. 154--173, 2000.

\bibitem{Gionis1999Similarity}
A.~Gionis, P.~Indyk, and R.~Motwani, ``Similarity search in high dimensions via
  hashing,'' in \emph{PVLDB}, 1999, pp. 518--529.

\bibitem{Weiss2008Spectral}
Y.~Weiss, A.~Torralba, and R.~Fergus, ``Spectral hashing,'' in \emph{Advances
  in Neural Information Processing Systems}, 2009, pp. 1753--1760.

\bibitem{Huang2015Query}
Q.~Huang, J.~Feng, Y.~Zhang, Q.~Fang, and W.~Ng, ``Query-aware
  locality-sensitive hashing for approximate nearest neighbor search,''
  \emph{PVLDB}, vol.~9, no.~1, pp. 1--12, 2015.

\bibitem{Liu2016Query}
X.~Liu, C.~Deng, B.~Lang, D.~Tao, and X.~Li, ``Query-adaptive reciprocal hash
  tables for nearest neighbor search.'' \emph{IEEE Transactions on Image
  Processing}, vol.~25, no.~2, pp. 907--919, 2016.

\bibitem{weber1998quantitative}
R.~Weber, H.-J. Schek, and S.~Blott, ``A quantitative analysis and performance
  study for similarity-search methods in high-dimensional spaces,'' in
  \emph{VLDB}, vol.~98, 1998, pp. 194--205.

\bibitem{jegou2011product}
M.~D. Jegou, Herve and C.~Schmid, ``Product quantization for nearest neighbor
  search,'' \emph{IEEE Transactions on Pattern Analysis and Machine
  Intelligence}, vol.~33, no.~1, pp. 117--128, 2011.

\bibitem{ge2014optimized}
T.~Ge, K.~He, Q.~Ke, and J.~Sun, ``Optimized product quantization,'' \emph{IEEE
  Transactions on Pattern Analysis and Machine Intelligence}, vol.~36, no.~4,
  pp. 744--755, 2014.

\bibitem{zhang2014composite}
T.~Zhang, C.~Du, and J.~Wang, ``Composite quantization for approximate nearest
  neighbor search.'' in \emph{ICML}, no.~2, 2014, pp. 838--846.

\bibitem{johnson2017billion}
J.~Johnson, M.~Douze, and H.~J{\'e}gou, ``Billion-scale similarity search with
  gpus,'' \emph{arXiv:1702.08734}, 2017.

\bibitem{arya1993approximate}
S.~Arya and D.~M. Mount, ``Approximate nearest neighbor queries in fixed
  dimensions,'' in \emph{SODA}, vol.~93, 1993, pp. 271--280.

\bibitem{Hajebi2011Fast}
K.~Hajebi, Y.~Abbasi-Yadkori, H.~Shahbazi, and H.~Zhang, ``Fast approximate
  nearest-neighbor search with k-nearest neighbor graph.'' in \emph{IJCAI
  2011}, vol.~22, 2011, pp. 1312--1317.

\bibitem{malkov2014approximate}
Y.~Malkov, A.~Ponomarenko, A.~Logvinov, and V.~Krylov, ``Approximate nearest
  neighbor algorithm based on navigable small world graphs,'' \emph{Information
  Systems}, vol.~45, pp. 61--68, 2014.

\bibitem{MalkovYHNSW16}
Y.~A. Malkov and D.~A. Yashunin, ``Efficient and robust approximate nearest
  neighbor search using hierarchical navigable small world graphs,'' \emph{IEEE
  Transactions on Pattern Analysis and Machine Intelligence}, 2018.

\bibitem{Ben2016Fanng}
H.~Ben and D.~Tom, ``{FANNG}: Fast approximate nearest neighbour graphs,'' in
  \emph{Proceedings of the 2016 IEEE Conference on Computer Vision and Pattern
  Recognition}, 2016, pp. 5713--5722.

\bibitem{fu2019fast}
C.~Fu, C.~Xiang, C.~Wang, and D.~Cai, ``Fast approximate nearest neighbor
  search with the navigating spreading-out graph,'' \emph{{PVLDB}}, vol.~12,
  no.~5, pp. 461--474, 2019.

\bibitem{aumuller2017ann}
M.~Aum{\"u}ller, E.~Bernhardsson, and A.~Faithfull, ``Ann-benchmarks: A
  benchmarking tool for approximate nearest neighbor algorithms,'' in
  \emph{International Conference on Similarity Search and Applications}.\hskip
  1em plus 0.5em minus 0.4em\relax Springer, 2017, pp. 34--49.

\bibitem{shimomura2020survey}
L.~C. Shimomura, R.~S. Oyamada, M.~R. Vieira, and D.~S. Kaster, ``A survey on
  graph-based methods for similarity searches in metric spaces,''
  \emph{Information Systems}, p. 101507, 2020.

\bibitem{dearholt1988monotonic}
D.~Dearholt, N.~Gonzales, and G.~Kurup, ``Monotonic search networks for
  computer vision databases,'' in \emph{Signals, Systems and Computers, 1988.},
  vol.~2.\hskip 1em plus 0.5em minus 0.4em\relax IEEE, 1988, pp. 548--553.

\bibitem{beckmann1990r}
N.~Beckmann, H.-P. Kriegel, R.~Schneider, and B.~Seeger, ``The {R}*-tree: an
  efficient and robust access method for points and rectangles,'' in
  \emph{{ACM} Sigmod Record}, vol.~19, no.~2.\hskip 1em plus 0.5em minus
  0.4em\relax Acm, 1990, pp. 322--331.

\bibitem{HuangFZFN15}
Q.~Huang, J.~Feng, Y.~Zhang, Q.~Fang, and W.~Ng, ``Query-aware
  locality-sensitive hashing for approximate nearest neighbor search,''
  \emph{{PVLDB}}, vol.~9, no.~1, pp. 1--12, 2015.

\bibitem{ge2013optimized}
T.~Ge, K.~He, Q.~Ke, and J.~Sun, ``Optimized product quantization for
  approximate nearest neighbor search,'' \emph{Proceedings of the IEEE
  Conference on Computer Vision and Pattern Recognition}, pp. 2946--2953, 2013.

\bibitem{Jin2014Fast}
Z.~Jin, D.~Zhang, Y.~Hu, S.~Lin, D.~Cai, and X.~He, ``Fast and accurate hashing
  via iterative nearest neighbors expansion,'' \emph{IEEE Transactions on
  Cybernetics}, vol.~44, no.~11, pp. 2167--2177, 2014.

\bibitem{CongEfanna2016}
C.~Fu and D.~Cai, ``Efanna : An extremely fast approximate nearest neighbor
  search algorithm based on knn graph,'' \emph{arXiv:1609.07228}, 2016.

\bibitem{MalkovPLK14}
Y.~Malkov, A.~Ponomarenko, A.~Logvinov, and V.~Krylov, ``Approximate nearest
  neighbor algorithm based on navigable small world graphs,'' \emph{Inf.
  Syst.}, vol.~45, pp. 61--68, 2014.

\bibitem{aurenhammer1991voronoi}
F.~Aurenhammer, ``Voronoi diagrams—a survey of a fundamental geometric data
  structure,'' \emph{ACM Computing Surveys (CSUR)}, vol.~23, no.~3, pp.
  345--405, 1991.

\bibitem{lee1980two}
D.-T. Lee and B.~J. Schachter, ``Two algorithms for constructing a delaunay
  triangulation,'' \emph{International Journal of Computer \& Information
  Sciences}, vol.~9, no.~3, pp. 219--242, 1980.

\bibitem{kleinberg2000navigation}
J.~M. Kleinberg, ``Navigation in a small world,'' \emph{Nature}, vol. 406, no.
  6798, pp. 845--845, 2000.

\bibitem{boguna2009navigability}
M.~Boguna, D.~Krioukov, and K.~C. Claffy, ``Navigability of complex networks,''
  \emph{Nature Physics}, vol.~5, no.~1, pp. 74--80, 2009.

\bibitem{jaromczyk1992relative}
J.~W. Jaromczyk and G.~T. Toussaint, ``Relative neighborhood graphs and their
  relatives,'' \emph{Proceedings of the IEEE}, vol.~80, no.~9, pp. 1502--1517,
  1992.

\bibitem{li2016approximate}
W.~Li, Y.~Zhang, Y.~Sun, W.~Wang, M.~Li, W.~Zhang, and X.~Lin, ``Approximate
  nearest neighbor search on high dimensional data-experiments, analyses, and
  improvement,'' \emph{IEEE Transactions on Knowledge and Data Engineering},
  2019.

\bibitem{Dong2011Efficient}
W.~Dong, C.~Moses, and K.~Li, ``Efficient k-nearest neighbor graph construction
  for generic similarity measures,'' in \emph{Proceedings of the 20th
  International Conference on World Wide Web}, 2011, pp. 577--586.

\bibitem{yao1982constructing}
A.~C.-C. Yao, ``On constructing minimum spanning trees in k-dimensional spaces
  and related problems,'' \emph{SIAM Journal on Computing}, vol.~11, no.~4, pp.
  721--736, 1982.

\end{thebibliography}




%
 \vspace{-0.8cm}
\begin{IEEEbiography}[{\includegraphics[width=1in,height=1.25in,clip,keepaspectratio]{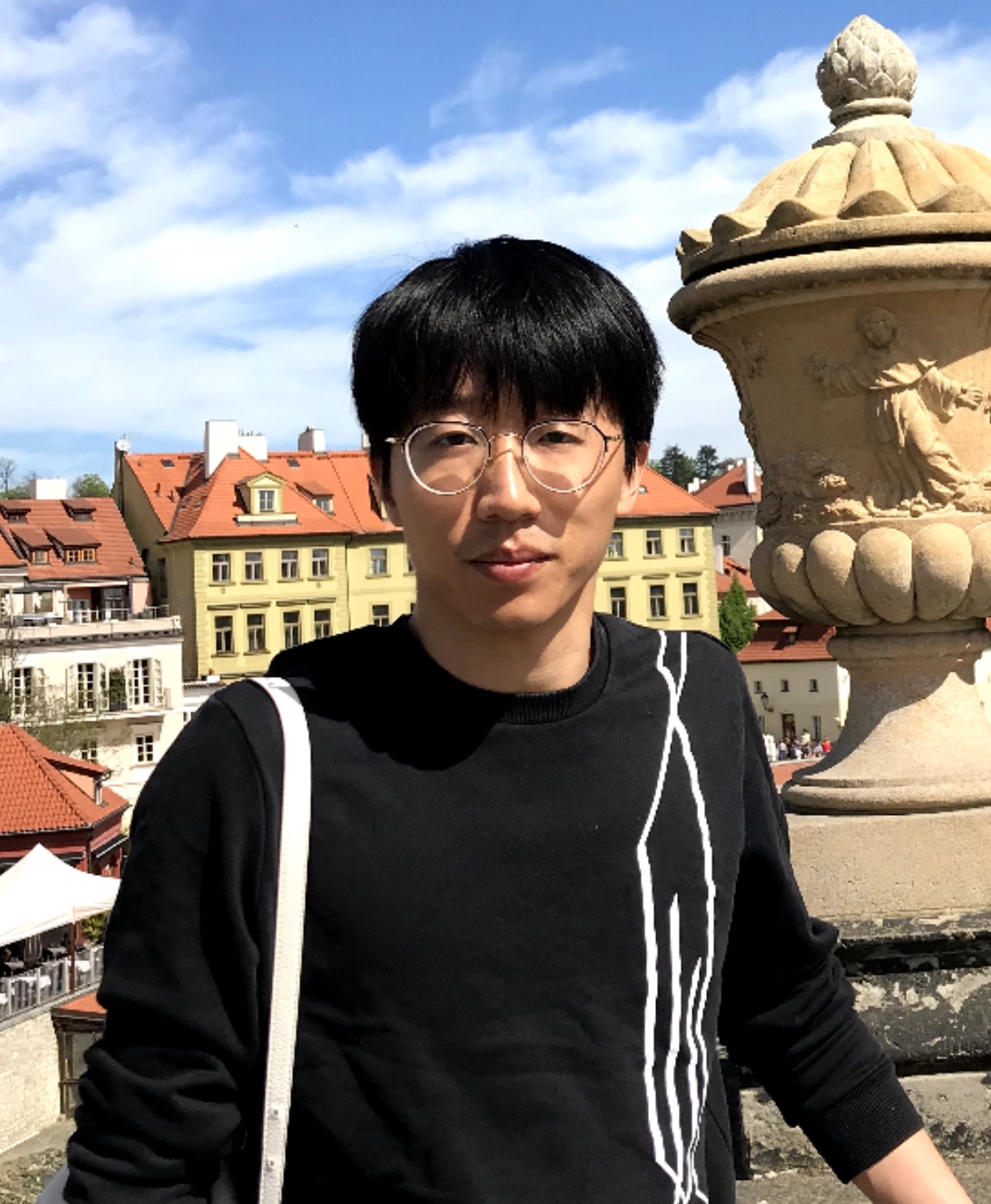}}]{Cong Fu}
 is a Ph.D. student in the State Key Lab of CAD\&CG, College of Computer Science at Zhejiang University, China. He received the bachelor degree in computer science from Zhejiang University in 2015. His research interests include recommendation system, knowledge graph, large-scale database, and information retrieval.
\end{IEEEbiography}
 \vspace{-0.8cm}
\begin{IEEEbiography}[{\includegraphics[width=1in,height=1.25in,clip,keepaspectratio]{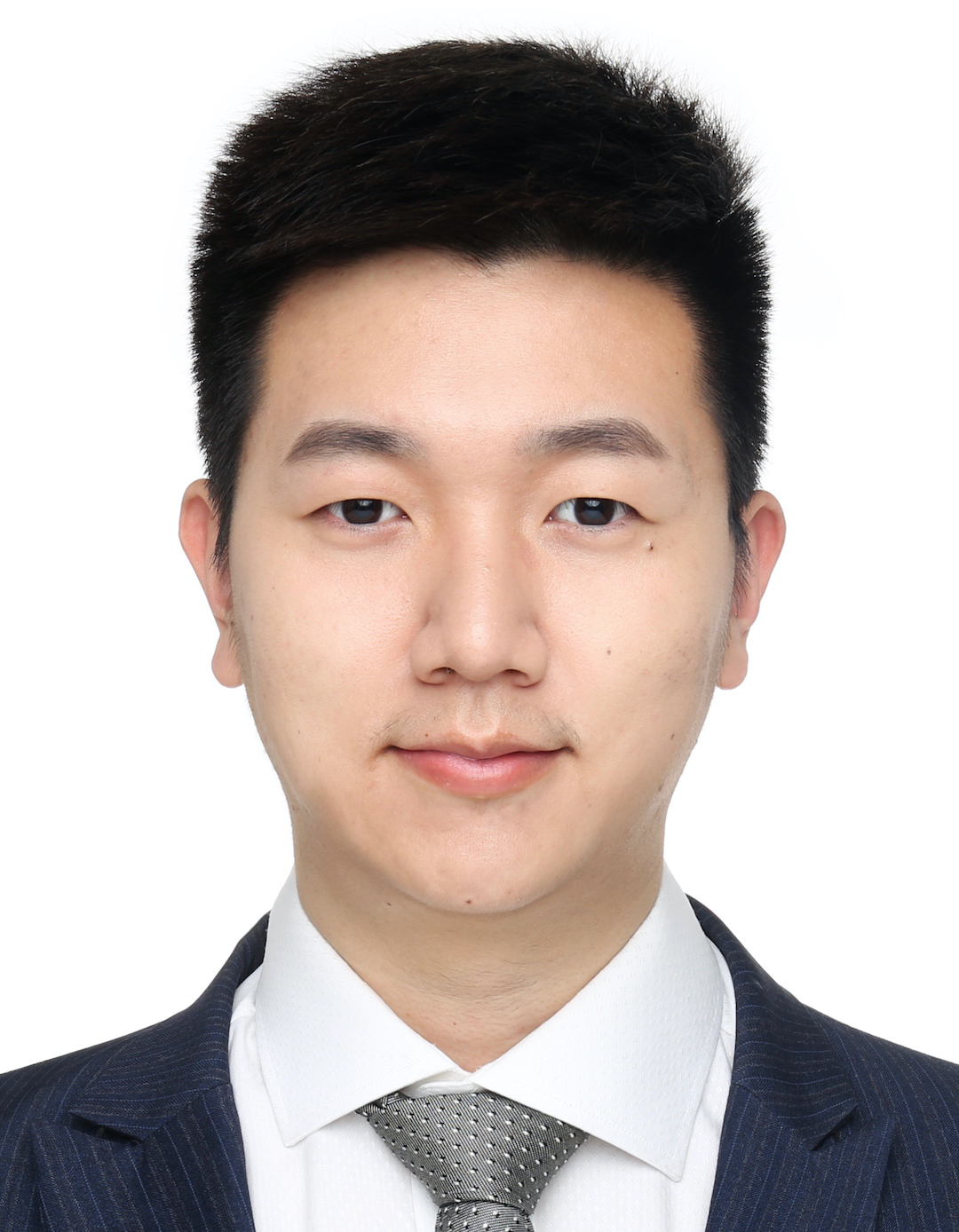}}]{Changxu Wang}
 works for Alibaba DAMO Academy for Discovery, Adventure, Momentum and Outlook. He received the master degree in computer science from Zhejiang University in 2019. His research interests include machine learning platform, computer vision, and information retrieval. 
\end{IEEEbiography}
 \vspace{-0.8cm}
\begin{IEEEbiography}[{\includegraphics[width=1in,height=1.25in,clip,keepaspectratio]{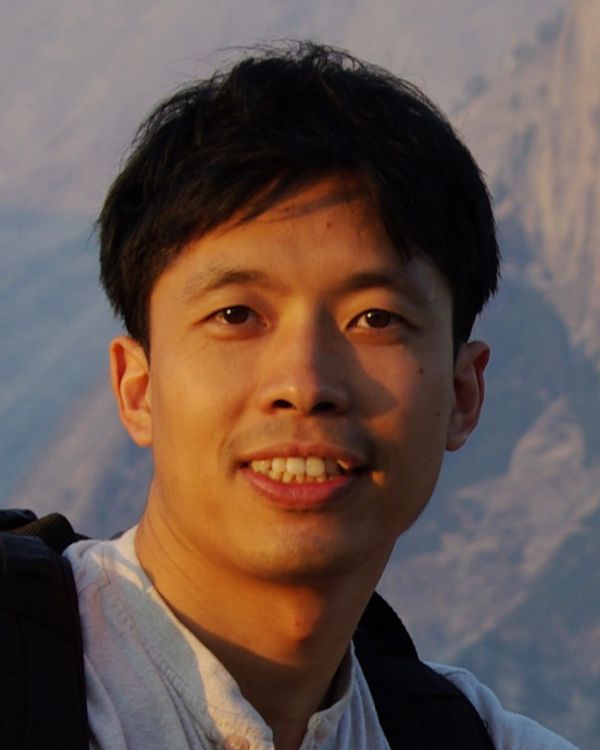}}]{Deng Cai}
 is a Professor in the State Key Lab of CAD\&CG, College of Computer Science at Zhejiang University, China. He received the Ph.D. degree in computer science from University of Illinois at Urbana Champaign in 2009. His research interests include machine learning, data mining and information retrieval. He is a member of IEEE.
\end{IEEEbiography}



\clearpage
\appendices
\section{Relationship Between NSG, RNG*, and This Work}
NSG \cite{fu2019fast}, RNG* \cite{arya1993approximate} is closely related to this work. Considering that they are not well-known methods, and we try to make the main body of this paper focus more on SSG, we will briefly introduce them here and discuss the relationships between them and this work, especially the differences.

\subsection{A Brief Introduction to NSG paper}
Fu~\etal \cite{fu2019fast} make efforts to find a new graph index, specifically they hope the out-degree of the graph to be as low as possible, and meanwhile the path connecting any two nodes in the graph can be as short as possible. Different from previous methods which are mostly based on heuristics, Fu~\etal~\cite{fu2019fast} try to find a basic theoretical framework and develop a new index from it. In particular, they first find out an early graph model MSNET \cite{dearholt1988monotonic} and provide a few theoretical analysis of MSNET. MSNET defines a group of graphs. Fu~\etal~\cite{fu2019fast} propose MRNG and prove that MRNG is a kind of MSNET. To lower the indexing complexity, they propose an approximation model, namely NSG. Empirical results show that NSG is very sparse and meanwhile efficient in search. Consequently, NSG uses the least memory compared with other graph-based methods and achieves the best search performance.

\subsection{Differences Between NSG and This Work}
It is not difficult to find out that the paper structure of this work is similar to that of NSG \cite{fu2019fast}: proposing a theoretical model and develop an approximation for practical use. 

However, this work and NSG are two completely different methods. The MSNET \cite{dearholt1988monotonic} is proposed by Dearholt~\etal about 30 years ago, which is the basic theoretical framework of both NSG and this work. The only common point of NSG and this work is that MRNG proposed in their work and SSG in this work both belong to the MSNET family. The differences between NSG and this work are as follows:

\begin{enumerate}
\item MRNG and SSG are MSNETs, but their ideas are completely different. Fig. \ref{rule_mrng_ssg} is an illustration of the different edge selection strategy. MRNG tries to prune as much long edges as possible as long as it does not affect the connectivity of the graph. Therefore, in a) of Fig. \ref{rule_mrng_ssg}, $ac$ cannot be an out-edge of MRNG, because $ac$ is the longest edge in triangle $abc$ and $ab$ has been selected as an out-edge in the MRNG. In contrast, SSG cares more about the out-direction coverage. Thus, SSG uses the angle between edges as the edge selection criteria. 

\item The theoretical properties of MRNG and SSG are different. MRNG only has theoretical guarantees for in-database queries, but SSG has theoretical guarantees for both in-database queries and not-in-database queries. 

\item Under the general position assumption (there is no isosceles triangle formed by any node combination, commonly used in graph theory), an MRNG defined on a given dataset is unique. But there are a lot of SSGs which can be built on a given dataset, because each $\alpha, (0\le \alpha \le 60^{\circ})$ defines an SSG. To some extent, the sparsity of SSGs can be easily  and precisely manipulated.

\item The indexing pipeline of NSG (MRNG approximation) and NSSG (SSG approximation) are completely different. Simply speaking, building NSG follows a search-collect-prune pipeline; building NSSG follows a propagation, pruning, and reverse-neighbor linking pipeline. Thereby, in NSG indexing, edge processing dominates the indexing complexity (in O($n\log n$) scale), but KNN graph construction dominates the NSSG indexing (edge processing is of O($n$) complexity). 
\end{enumerate}
\begin{figure}[t]
\begin{center}
\includegraphics[width=0.48\textwidth]{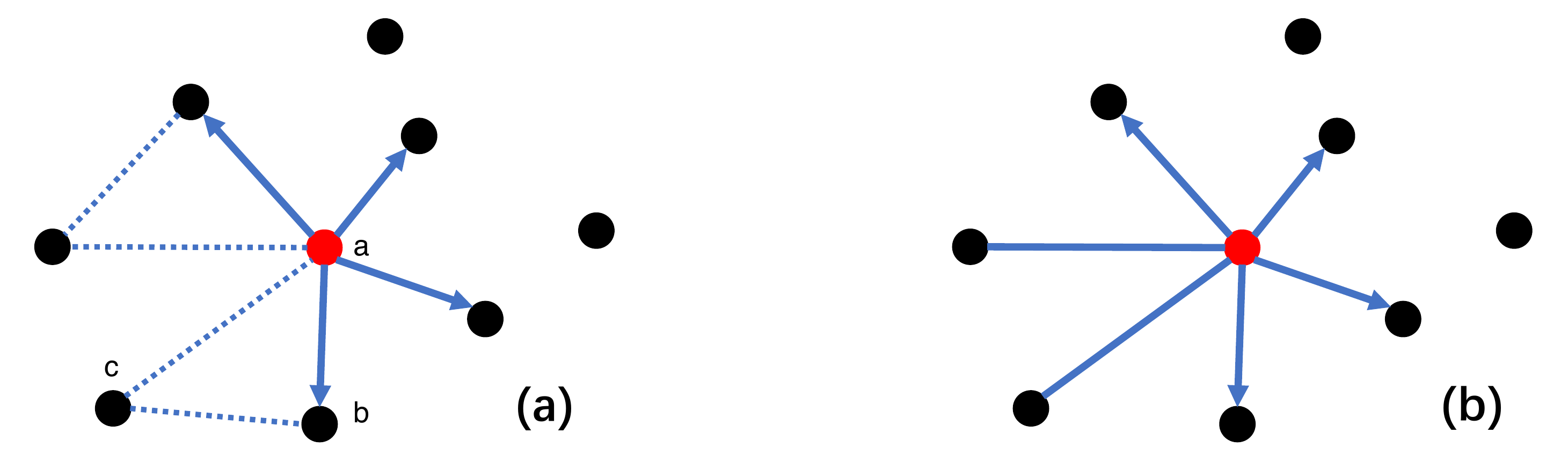}
\end{center}
   \caption{\textbf{Edge distribution difference between MRNG and SSG.} a) is a typical MRNG. b) is an SSG. }
\label{rule_mrng_ssg}
\end{figure}
In summary, SSG has better theoretical guarantees than MRNG; The low indexing complexity ensures NSSG higher scalability than NSG; The flexible sparsity manipulation ensures SSG a better adaptability on all kinds of data distribution. 

\subsection{A Brief Introduction to RNG* paper}
Arya~\etal\cite{arya1993approximate} propose a new graph-based ANNS method named as Randomized Neighborhood Graph (RNG*, distinguish from Relative Neighborhood Graph). RNG* can be viewed as a graph pruned from a complete graph. For a given point $p$ in the graph, its neighbors are selected with the following routine. The space around $p$ is divided into adjacent cones. Then the points are grouped by cones. The points in each cone are given unique numbers as IDs. The points are sorted in ascending order of the IDs. Next the points are checked from the beginning one by one. If a point is the closest to $p$ among the points of lower indices in this cone, it will be selected. This produces a graph with a O($\log^3 n$) search complexity. 

Though theoretically attractive, RNG* suffers from high indexing cost, high memory cost, and high search cost in high dimensions. Therefore they propose a variant named as RNG*(S). RNG*(S) can be also viewed as a graph pruned from a complete graph. For a given point $p$, its neighbors are selected with the following routine. The rest points are sorted in ascending order of the distance to $p$. The nearest one, some point $r$, is first selected. If a point $s$ satisfies that $\delta(p,s) < \delta(s,r)$, $s$ is selected. We get RNG*(S) when all points are pruned in this way.

\subsection{Differences Between RNG* and This Work}
Intuitively, in terms of pruning strategy, MRNG is similar to RNG*(S), and SSG is similar to RNG*. The differences are as follows.

\begin{enumerate}
\item During edge pruning, SSG is generated in a deterministic way, but RNG* involves randomness. SSG selects only one neighbor in each cone for a given point, while RNG* select O($\log n$) neighbors in each cone. Consequently, SSGs are much sparser than RNG*. SSG has a clear discussion on the hyper-parameter $\alpha$, while RNG* does not tell how to select a proper angle diameter for each cone.

\item SSG considers both angle and distance in pruning, while RNG*(S) only considers distances. MRNG and RNG*(S) are the same things.
\end{enumerate}

\section{Theoretical Analysis}

\subsection{Proof For Theorem \ref{theorem1}}

\begin{proof}
According to the Lemma 1 in the NSG paper \cite{fu2019fast}, in an Euclidean space $E^d$, a graph $G$ defined on a finite point set $S$ is monotonic if and only if for any two unconnected points $p,q \in S$, there is at least one edge $\overset{\longrightarrow}{pr}$ such that $r \in B(q, \delta(p,q))$. $B(q, \delta(p,q))$ is defined as the open sphere centered at $q$ with radius $\delta(p,q)$. 

For any two node $p,q \in S$, there are only two kinds of relationships between them according to the definition of SSG.

1) $p,q$ is connected. If $p,q$ is connected, the path between $p,q$ contains only one edge, which is monotonic. 

2) $p,q$ is not connected. If $p,q$ is not connected, there must be at least one edge $\overset{\longrightarrow}{pr} \in SSG$ such that $\overset{\longrightarrow}{pr} \in Cone(\overset{\longrightarrow}{pq}, \alpha)$ and $\delta(p,r) < \delta(p,q)$. Because $\angle{rpq} < \alpha \le 60^{\circ}$, $rq$ is not the longest edge in $\triangle{pqr}$, and we have $\delta(r,q) < \delta(p,q)$. Thus it is monotonic towards $q$ when going from $p$ to $r$. This process can be repeated iteratively until we have a monotonic search path from $p$ to $q$. Finally, we have SSG is monotonic. 

Since SSG is monotonic, according to the definition of Monotonic Search Network \cite{fu2019fast}, SSG is strongly connected.
\end{proof}
\begin{figure}[!ht]
\begin{center}
\includegraphics[width=220pt]{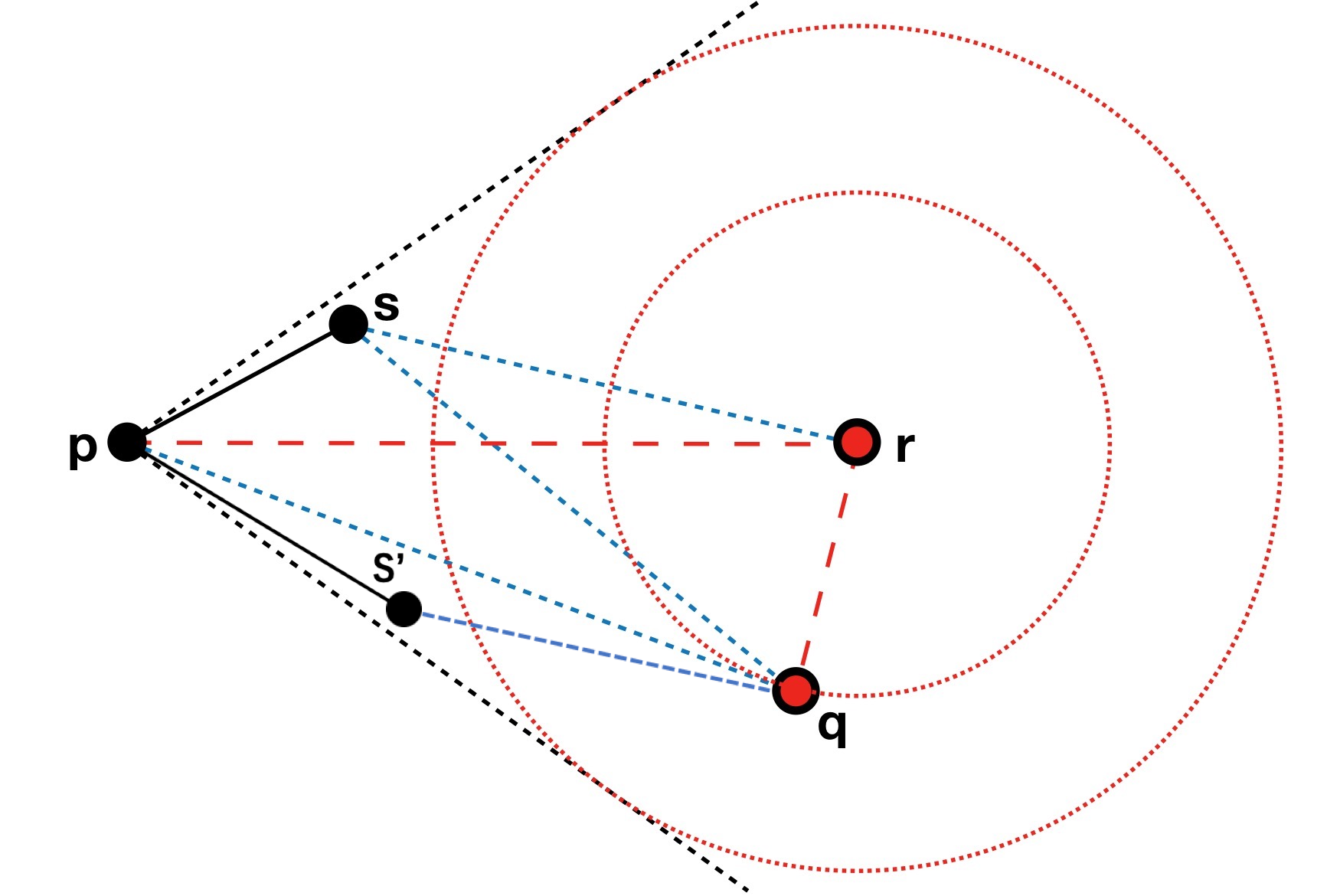}
\end{center}
   \caption{An illustration of the proof for the partially monotone path towards the query which is not in the database. Note that $p,q,r,s$ may not be in the same hyper-plane.}
\label{comonotone}
\end{figure}
\subsection{Proof For Theorem \ref{theorem2}}

First we need to prove that the upper-bound of SSG's degree is a constant. 

\begin{lemma}
\label{lemma_degree}
In an Euclidean space $E^d$, for any $0 < \alpha < \pi$, there are finite edges sharing a common end node such that the angle between any two of them is $\beta$ and $\beta \ge \alpha$. The supremum of the number of such edges is a function of $\alpha$ and $d$.
\end{lemma}
\begin{proof}
Yao\cite{yao1982constructing} prove that, for any $0 < \alpha < \pi$, one can cover the space $E^d$ with finite convex cones such that $sup\{\Theta(C_i)| C_i \in \mathbb{C}\}= \varphi$. Let $\mathbb{C} = \{C_1, C_2, ..., C_k\}$ be a set of convex cones constructed by the algorithm proposed by Yao and with $\varphi = \frac{1}{2}\alpha - \varepsilon$, where $\varepsilon$ is a very small value and $\varepsilon > 0$. Let $u$ be the center node and $uv, uw$ be any two edges. We define two convex cones $C_i, C_j$ to be \textit{adjacent} if $\exists{ua}, ua \in (C_i \cap C_j)$.

We will first prove that the closest two edges (two edges forming the smallest angle) must be inside the same cone or two adjacent cones. Let $\angle{abc}$ be the angle between $ba$ and $bc$. If $uv$ and $uw$ are in two cones, $C_i$ and $C_j$, that are not adjacent, \ie, $C_i \cap C_j = \emptyset$, then $inf\{\angle{vuw} | uv \in C_i, uw \in C_j\} = \gamma > 0$. If $uv$ and $uw$ are in two adjacent cones, it's obvious that $inf\{\angle{vuw}\} = 0$. Thus, the closest two edges must appear in the same cone or two adjacent cones.

If $uv$ and $uw$ are inside of the same cone, then $\angle{vuw} \le \varphi < \alpha$ because $sup\{\Theta(C_i)| C_i \in \mathbb{C}\}= \varphi$. Suppose $uv$ and $uw$ are inside of any two \textit{adjacent} convex cones. Let $uz$ be an edge, $uz \in (C_i \cap  C_j)$. Because $\angle{vuz} \le \varphi$ and $\angle{wuz} \le \varphi$, $uv$ and $uw$ is inside the  circular cone $\widetilde{C}(\varphi$, \textit{\textbf{uz}}), and $sup\{\angle{vuw}\} = 2\varphi = \alpha - 2\varepsilon < \alpha$. Suppose we select any $k$ edges sharing a common end node and place them in $\mathbb{C}$. Given that $|\mathbb{C}| = k$, if there exists one empty cone, there must be two edges within the same cone. If we place one edge in each cone, the angle between any two edges is smaller than $\alpha$. Therefore, there are at most $m$ edges sharing a common end node such that the angle between any two of them is $\beta$, where $\beta \ge \alpha$ and $m < k$. $m$ is a function of $d$ and $\alpha$, because they are the only variables involved.
\end{proof}
Lemma \ref{lemma_degree} indicates that when dimension and angle are fixed, the degree will be upper-bounded by a constant. Then we can get the search complexity on SSG for in-database queries.
\begin{proof}
According to the Theorem 2 in \cite{fu2019fast}, the length of the search paths for in-database queries is $O(n^{1/d}$ $\log(n^{1/d})/\triangle{r})$, where $n$ is the size of the dataset, and $\triangle{r}$ is the lower-bound of lengths of edges in any non-isosceles triangles. Suppose $D$ is the upper-bound of the degree of SSG, we can get the search complexity for in-database queries, $O(Dn^{1/d}\log(n^{1/d})/\triangle{r})$, according to the Theorem 1 in \cite{fu2019fast}.
\end{proof}

\subsection{Proof For Theorem \ref{theorem3}}

\begin{proof}
Figure \ref{comonotone} is an illustration of this proof. The search algorithm discussed here is Alg. 1 mentioned in the main body of this paper. The search start point is selected randomly. Suppose $p$ is any point on the search path produced by Alg. 1. Under the condition in this theorem, we assume the current point $p$ can its neighbors are all far from the query and its nearest neighbor. Let $q$ denote the query which is not indexed in the database, and $r$ is $q$'s nearest neighbor which is within the database. To avoid ambiguity, if there are more than one nearest neighbors (two neighbors share the same distance to $q$), we only consider return at least one right answer. We can give all the nodes unique IDs (e.g. indexing them from 1 to $n$) to distinguish them, and here we refer to the nearest neighbor $r$ with a smaller ID. 

According to the definition of SSG, we have at least one edge $\overset{\longrightarrow}{ps}$ such that $\overset{\longrightarrow}{ps} \in Cone(\overset{\longrightarrow}{pr}, \alpha) \cap B(p,\delta(p,r)) \cap S$.  Suppose plane $\Phi$ is the normal plane of the hyper-plane determined by point $q,p,r$ and $pr \in \Phi$. We can get plane $\Phi$ divide $Cone(\overset{\longrightarrow}{pr}, \alpha)$ into two parts evenly. Because the data are distributed randomly in the space, the probability of $s$ to be within each half of the cone is 0.5.

1) If $s$ is within the same half as $q$ (like the position of $s'$), we can have $\angle{s'pq} < \alpha \le 60^{\circ}$, and $s'q$ is not the longest edge in $\triangle{pqs'}$. Because $\forall \overset{\longrightarrow}{ps'} \in SSG, \delta(p,s') < \delta(p,q)$, we have $pq$ is the longest edge in $\triangle{pqs'}, \delta(p,q) > \delta(s'q) $. According to the definition of monotonicity \cite{fu2019fast}, the search process is monotonic towards $q$. Because $s'$ is also on the monotonic path to $r$ according to the property of SSG, the search process is monotonic towards both $q$ and $r$. 

2) If $s$ is within the other half, the search process is monotonic towards both $q$ and $r$ if $\angle{spq} < 60^{\circ}$, \ie, $\angle{spr} + \angle{qpr} < 60^{\circ}$. The probability of such an event is $\epsilon, 0 < \epsilon \le 0.5$. So the total probability that the search process is monotonic towards both $q$ and $r$ is $0.5+\epsilon$. It is easy for one to verify that if $\alpha < 30^{\circ}$, $\angle{spq} < 60^{\circ}$ definitely, and $0.5 + \epsilon=1$.
\end{proof} 

\section{Experiment on Approximation}
Though there is not theoretical connection between SSG and NSSG in this paper, the performance of NSSG benefits from that it do approximate SSG ideal structure to some extent. We can show it by a simple experiment as follows.

\begin{table}[t]
	\caption{Results of the experiment on how much NSSG approximates SSG. 1-NN acc. and 50-NN acc. denotes the 1-NN and 50-NN accuracy of the approximate K-NN graph. Edge overlap ratio denotes the fraction of shared edges in all NSSG edges. Degree denotes the degree of NSSGs.}
	\label{approximation_exp}
	\centering
	
	\begin{tabular}{cccc}
		\hline
		1-NN acc. & 50-NN acc. & edge overlap ratio & degree\\
		\hline
		0.999 & 0.981 & 0.997 & 50 \\
	
	    	0.951 & 0.934 & 0.949 & 50 \\
		
		0.901 & 0.880 & 0.893 & 50 \\
		
		0.855 & 0.811 & 0.848 & 50 \\
		\hline
	\end{tabular}
\end{table}

On the same dataset SIFT10K, we build multiple approximate K-NN graphs with different accuracy. Subsequently, we build an SSG with same parameter configurations on each K-NN graph. We count how many edges each NSSG and the exact SSG share. The results can be found in Table \ref{approximation_exp}.

The K-NN (K=50) graphs are built with nn-descent \cite{Dong2011Efficient}. The 1-NN acc. is calculated according to that how many exact nearest neighbor are connected in this approximate K-NN graph compared with an exact K-NN graph. The 50-NN acc. is calculated similarly. From the results we can see that NSSG share a lot of edges with exact SSG, which is mainly determined by the 1-NN accuracy of the approximate K-NN graph. According to our experience, NSSG presents better search performance built on more accurate K-NN graphs, i.e., better approximates the exact SSG.

\end{document}